\documentclass[11pt, letterpaper]{article}

\usepackage{games}
\usepackage{verbatim}

\usepackage{tikz}
\begin{document}

\title{Approximating Nash Equilibrium in Random Graphical Games}
\author{ Morris Yau \thanks{ \texttt{morrisy@mit.edu}}}
\affil{Massachusetts Institute of Technology}
\date{\today}
\maketitle

\begin{abstract}
       Computing Nash equilibrium in multi-agent games is a longstanding challenge at the interface of game theory and computer science.  It is well known that a general normal form game in $N$ players and $k$ strategies requires exponential space simply to write down.  This \emph{Curse of Multi-Agents} prompts the study of \emph{succinct} games which can be written down efficiently.  A canonical example of a \emph{succinct} game is the \emph{graphical game} which models players as nodes in a graph interacting with only their neighbors in direct analogy with markov random fields.  Graphical games have found applications in wireless, financial, and social networks.  However, computing the nash equilbrium of graphical games has proven challenging. Even for \emph{polymatrix} games, a model where payoffs to an agent can be written as the sum of payoffs of interactions with the agent's neighbors, it has been shown that computing an $\epsilon$ approximate nash equilibrium is PPAD hard for $\epsilon$ smaller than a constant.  The focus of this work is to circumvent this computational hardness by considering average case graph models i.e random graphs.  We provide a quasipolynomial time approximation scheme (QPTAS) for computing an $\epsilon$ approximate nash equilibrium of polymatrix games on random graphs with edge density greater than $\poly(k, \frac{1}{\epsilon}, \ln(N))$ with high probability.  Furthermore, with the same runtime we can compute an $\epsilon$-approximate Nash equilibrium that $\epsilon$-approximates the maximum social welfare of any nash equilibrium of the game.  Our primary technical innovation is an "accelerated rounding" of a novel hierarchical convex program for the nash equilibrium problem.  Our accelerated rounding also yields faster algorithms for Max-2CSP on the same family of random graphs, which may be of independent interest.  
       
\end{abstract}

\thispagestyle{empty}
\setcounter{page}{0}
\newpage

\tableofcontents

\thispagestyle{empty}
\setcounter{page}{0}
\newpage

\section{Introduction}

Beginning with the seminal works of Nash \cite{Nash48} and Von Neumann \cite{vonneumann1947}, game theory has been concerned with the study of strategic interactions amongst rational agents. In his landmark 1951 paper, Nash defined his namesake equilibrium and proved its existence for every game in mixed strategies.  Since the 1950's, the rise of the internet heralded a new brand of commerce for computational ecosystems comprised of many agents.  Today, the applications of game theory span a remarkable range from cryptocurrency and blockchain technologies to the deployment and design of 5G networks and edge devices.  These developments brought computation to the forefront of game theory, wherein the problem of computing the Nash equilibrium of two player games came to be of paramount importance.  

Towards resolving this problem in the negative, Daskalakis, Goldberg, and Papadimitriou \cite{DP06} in their oft celebrated work proved that computing two player nash equilibrium is PPAD-hard.  As is standard in theoretical computer science, the subsequent hope is to approximate the nash equilibrium to additive error $\epsilon$.  These efforts have proven fruitful, as a quasipolynomial time approximation scheme was achieved by Lipton, Mehta, and Markakis \cite{LMM03}.  Despite the good news, extending these techniques to larger numbers of players is met with an immediate obstacle.  

Quite simply, the normal form game between $N$ agents taking $k$ actions requires $k^N$ space to write down. Even linear time algorithms would have runtimes exponential in $N$.  Thus, an algorithmic theory of multi-agent games would have to be married with succinct models of multi-agent interactions.  A particularly attractive modeling choice is that of the \emph{graphical game} introduced by Kearns, Littman, and Singh \cite{KLS13}.  \emph{Graphical games} are a class of games where agents are represented as nodes in a graph interacting with their neighbors.  At one extreme, every game can be modeled by interactions on the complete graph which offers no savings in space. At the other end of the spectrum, games played on sparse graphs can be succinctly represented and capture settings where agents are influenced by a comparatively small set of neighbors.  This is especially relevant in network settings where devices compete for bandwidth resources with their proximal neighbors; in financial networks where institutions transact with peer institutions; in social networks where agents can influence a small set of trusted associates.  As such, \emph{graphical games} are practically well motivated models for mitigating what is colloquially known as the \emph{Curse of Multi-agents}.         

Of course, the elephant in the room is whether Nash equilibrium can be efficiently computed in graphical games.  The answer to this question ought to depend on the structure of the game.  However, even for polymatrix games, which are graphical games where the payoff to each agent is the sum of the payoffs of games played with the agent's neighbors in a graph, it is PPAD hard to compute a nash equilibrium.  Even worse for the algorithmist, Rubinstein \cite{Rubinstein14a} showed that it is PPAD hard to $\epsilon$-approximate the nash equilibrium of a polymatrix game for small constant $\epsilon > 0$.  Moreoever, the hard instance is a simple $3$-regular graph.  This astounding negative result stands in stark contrast with the good news enjoyed by approximation algorithms in two player games.  In light of strong negative results, the next natural step in the theoretical computer science toolbox is to consider average case models for games.  In this work, we aim to answer the question\\

\textit{Does there exist a large class of average case games on graphs, for which it is possible to design efficient}
\centerline{\textit{algorithms for approximating nash equilibrium?}}
\\
\\
At a high level, our answer to this question is that for polymatrix games with smooth payoffs on random graphs that are mildly dense, approximating a nash equilibrium is as easy as approximating the nash equilibrium of a two player game.  That is to say, there exists a quasipolynomial time approximation scheme.  The intuitive reason for this is that for a random graph with smooth payoffs, the effect of an agents neighbors can be captured by succinct aggregate statistics of their actions reminiscent of both mean field theory in statistical physics and the study of maximum constraint satisfaction problems in approximation algorithms.  It is this latter family of problems from which we draw technical inspiration.  

\section{Preliminaries}
We begin by introducing some notation and defining the polymatrix game.  We will also define what it means for a payoff function to be $L$-smooth.
\paragraph{Notation:}  We work with $N \in \Z^+$ player games where each player can adopt one of $k \in \Z^+$ strategies.  For a graph $G$ let $V$ and $E$ be the vertices and edges of the graph respectively.  For a node $i \in V$ let $N(i)$ denote the set of vertices that are neighbors of $i$.  Let $a_i \in [k]$ denote the strategy of player $i$.  Let $\vec{a} \in [k]^N$ denote the vector of strategies $(a_1,a_2,...,a_N)$ adopted by all $N$ players.  Let $\calP(k,V)$ be the space of probability distributions over $[k]^N$.  Let $\Gamma \defeq \{\{\mu_{ip}\}_{p \in [k]}\}_{i \in [N]}$ denote a product probability distribution over $[k]^N$ where $\mu_{ip} = \mathbb{P}(a_i = p)$.  The focus of this work will be on random graphs drawn from $G(N,p)$ where $p$ is the probability of an edge existing between any pair of vertices.  We denote the average degree $d \defeq pN$.  
\begin{definition} (Polymatrix Game)
A polymatrix game is specified by a graph $G = (V,E)$ over $N$ vertices denoted $V$ and edges $E$. Furthermore, a set of payoff games $\{p_{ij}\}_{(i,j) \in E}$ is defined on each edge of $G$.  Each payoff game $p_{ij}: [k]^2 \rightarrow \R$ takes as input a pair of actions $(a_i,a_j)$ each selected from a space of $k$ discrete actions and outputs a real valued payoff.  The payoff function to player $i$ is $p_i(a_1,a_2,...,a_N) = \sum_{j \in N(i)} p_{ij}(a_i,a_j)$.  Note that  $p_{ij}(a_i,a_j)$ is the payoff to node $i$ and $p_{ji}(a_j,a_i)$ is the payoff to node $j$.            
\end{definition}
To enable a fair comparison for $\epsilon$ approximate nash equilibrium, we will rescale the payoffs so that the payoff to each agent is in the interval $[0,1]$.  See appendix for rescaling.  Henceforth we will be using $\{f_{ij}\}_{(i,j) \in E}$ as the rescaled payoff functions. Next we define $L$-smooth payoff functions.    

\begin{definition} \label{def:smoothness} (L-Smoothness)
We say a polymatrix game is $L$-smooth for a constant $L > 0$ if for all nodes $i$ and $j \in [N]$, the payoff to node $i$ for the game $f_{ij}$ is upper bounded  
\[\max_{p,q \in [k]^2} |f_{ij}(p,q)| \leq \frac{L}{d}\]
\end{definition}
For example, if each payoff $p_{ij}(\cdot, \cdot) \in [0,1]$ and $\min_{\vec{a} \in [k]^N} p_i(a_1,...,a_N) = 0$ and $\max_{\vec{a} \in [k]^N} p_i(a_1,...,a_N) = d$  then the rescaled payoff $f_{ij}(a_i,a_j) = \frac{p_{ij}(a_i,a_j)}{d}$ is $L$-smooth for $L = 1$.  Next we define the mixed nash equilibrium.  

\begin{definition} (Mixed Nash Equilibrium)
Let $\Gamma \defeq \{\{\mu_{ip}\}_{p \in [k]}\}_{i \in [N]}$ define a probability distribution over $[k]$ for all agents $i \in [N]$ where $\mathbb{P}_\Gamma(a_i = p) = \mu_{ip}$.  A mixed Nash equilibrium is distribution over mixed strategies $\Gamma$ satisfying the following equilibrium condition.

$$\mathbb{E}_{\Gamma}[f_i(a_1,...,a_n)] \geq \E_{\Gamma}[f_i(a_1,...,a_{i-1},w,...,a_n)]$$  
For all $w \in [k]$ and for all $i \in [n]$.
\end{definition}

\begin{definition} (Approximate Mixed Nash Equilibrium)
  We define an $\epsilon$-approximate Nash equilibrium as a product distribution over mixed strategies $\Gamma \defeq \{\{\mu_{ip}\}_{p \in [k]}\}_{i \in [N]}$ satisfying
$$\mathbb{E}_{\Gamma}[f_i(a_1,...,a_n)] \geq \E_{\Gamma}[f_i(a_1,...,a_{i-1},w,...,a_n)] - \epsilon$$  
For all $w \in [k]$ and for all $i \in [n]$.  
\end{definition}
 
Without the smoothness assumption and for general graphs $G$, it is known that computing an $\epsilon$-approximate nash equilibrium is PPAD hard for a small constant $\epsilon$.  In fact that hard instance is a three regular graph.  This paints a bleak picture for efficient equilbrium computation in possibly the most amenable class of graphical games.  In this work we find that for polymatrix games with $L$-smooth payoffs defined on random graphs with density greater than $\frac{k^8}{\epsilon^4} \ln^3(N)$, there exists a $(Nk)^{O(\frac{L^4\ln(k)}{\epsilon^4})}$ time algorithm for finding an $\epsilon$-approximate nash equilibrium with high probability $1 - N^{-\ln^2(N)}$ over the randomness in the graph.  

\begin{restatable} [Polymatrix Nash Equilibrium]{thm}{maintheorem} 
\label{thm:main-theorem}
Let $(G, \{p_{ij}\}_{i,j \in [N]})$ be a polymatrix game defined on a $G$ with rescaled and normalized payoffs $\{f_{ij}\}_{i,j \in [N]}$ that are $L$-smooth.  If $G$ is a random graph with edge density $d = \frac{L^8k^8}{\epsilon^8} \ln^3(N)$.  Then there is a $(Nk)^{O(\frac{L^4\ln(k)}{\epsilon^4})}$ time deterministic algorithm for finding an $\epsilon$-approximate nash equilibrium with high probability $1 - N^{-\ln^2(N)}$ over the randomness in the graph.     
\end{restatable}

\section{Related Work}
There are a variety of works on approximating polymatrix nash equilibria.  Polymatrix zero sum games are polynomial time solvable \cite{CCDP16}.  Polymatrix nash equilibrium can be approximated to a large additive constant \cite{DFSS14}.  Work on max-2CSP and combinatorial optimization via sum of squares include \cite{BRS11} \cite{RT12} \cite{GS13}. Finally, there is work on linear programming relaxations for computing nash equilibrium \cite{WHN16}.       

\section{A Convex Hierarchy for Nash Equilibrium}

\paragraph{Idealized Program: } Consider the following polynomial system $\calP$ comprised of polynomial inequalities in indeterminates $\{a_{ir}\}_{i \in [N], r \in [k]}$ where $a_{ir}$ is the indicator of the $r$'th action for player $i$.  A pure nash equilibrium is then a solution to the following polynomial system 
\[\calP =  \begin{cases} 
      \sum_{j \in N(i)} \sum_{p,q \in [k]^2}f_{ij}(p,q)a_{ip}a_{jq} - \sum_{j \in N(i)} \sum_{q \in [k]}f_{ij}(w,q)a_{jq} \geq 0 & \forall w \in [k] \\
      a_{ip}^2 - a_{ip} = 0 &  \forall i\in [N], p \in [k]\\
      \sum_{p=1}^k a_{ip} - 1 = 0 & \forall i \in [N] 
   \end{cases}
\]
Here the second constraint ensures $a_{ip} \in \{0,1\}$, which we refer to as the "booleanity" constraint.  The third constraint ensures each agent can adopt only one strategy, which we refer to as the "coloring" constraint.  The first constraint is the pure nash equilibrium constraint.  A pure nash equilibrium is not guaranteed to exist so $\calP$ is not necessarily feasible.  However, a mixed nash equilibrium is guaranteed to exist.  Therefore we know there exists a distribution over $[k]^N$ such that      
\[\calQ =  \begin{cases} 
      \E_\zeta[\sum_{j \in N(i)} \sum_{p,q \in [k]^2}f_{ij}(p,q)a_{ip}a_{jq} - \sum_{j \in N(i)} \sum_{q \in [k]}f_{ij}(w,q)a_{jq}] \geq 0 & \forall w \in [k] \\
      \E_\zeta[a_{ip}^2 - a_{ip}] = 0 &  \forall i \in [N],\forall p \in [k]\\
      \E_\zeta[\sum_{p=1}^k a_{ip} - 1] = 0 & \forall i \in [N]\\ 
      \E_\zeta[a_{ip}a_{jq}] = \E_\zeta[a_{ip}]\E_\zeta[a_{jq}] & \forall i \neq j \in [N], \forall p,q \in [k]^2 
   \end{cases}
\]

Here, $\calQ$ is a set of constraints over a distribution $\zeta$ over elements of $[k]^N$.  The first constraint it the mixed nash equilibrium constraint.  The second and third are the booleanity and coloring constraints.  The final constraint ensures that $\zeta$ is a product distribution which gives us a mixed strategy nash equilibrium. Indeed, a solution $\zeta$ to constraints in $\calQ$ would constitute a mixed Nash equilibrium.  However, finding $\zeta$ is evidently intractable.  For starters we are searching over the space of distributions over $[k]^N$, and the constraints are desperately non-convex.  Fortunately, there is a polynomial optimization toolbox for satisfying constraints that take on the form of $\calQ$.  Although an overview of polynomial optimization is beyond the scope of this paper, we will introduce some of the fundamentals and discuss how the case of Nash equilibrium requires conceptual modifications of the standard dogma.      

\paragraph{The Relaxation Toolbox and the Problem of Feasibility: }

In combinatorial optimization, one is often faced with maximizing a polynomial subject to polynomial constraints e.g maximum cut.  This reduces to solving systems of polynomial inequalities and equalities.  However, solving a polynomial system for an "integral" solution is NP-hard.  Therefore, we must settle for solving a convex relaxation of the original nonconvex polynomial optimization.  The convex relaxation solves for a distribution, or rather a pseudo-distribution, over integral solutions to the original polynomial optimization.  After designing and solving the convex relaxation, a rounding algorithm takes as input the pseudo-distribution (a distribution over large cuts) and aims to produce an approximate solution to the original problem e.g finding a large cut.  The key problem with designing a convex relaxation for Nash equilibrium is that the integral solution, the pure nash equilibrium, is not guaranteed to exist. This is more than a technicality, as both the Sherali Adams and Sum of Squares hierarchies of $\calP$ are infeasible.  

Alternatively, one may be motivated to relax the mixed nash equilibrium as it's guaranteed to exist and is also the object we are searching for.  However, this introduces a layer of misdirection that is hard to conceptually reason about.  If we're searching for a distribution over $[k]^N$, then the relaxation is attempting to find a pseudo-distribution over distributions over $[k]^N$.  Such an approach is hard to reason about and it is not the one taken in this work.  Rather, we search for mixed nash equilibrium by identifying a simple fact about correlated equilibria.

\paragraph{Conditioning Preserving Correlated Equilibrium}

A correlated equilibrium is a distribution $\zeta$ over $[k]^N$. An action vector $\vec{a} = (a_1,a_2,...,a_N) \in [k]^N$ is drawn from a distribution $\zeta$ such that each agent $i \in [N]$ plays action $a_i$. The distribution $\zeta$ satisfies that each agent $i$ has no incentive to deviate from its prescribed action $a_i \in [k]$ given that all other agents take actions drawn from $\zeta$ conditioned on $a_i$ denoted $\zeta|_{a_i}$.
The key insight is the following fact.   

In an $L$-smooth polymatrix game on a $G(N,p)$ random graph for $p \geq \frac{L^8k^8}{\epsilon^8} \ln^3(N)$, any correlated equilibrium that remains a correlated equilibrium after conditioning on all possible actions of all possible small subsets of agents is approximately a mixed nash equilibrium over the agents whose actions are not conditioned upon.

At a high level, a correlated equilibrium is not a Nash equilibrium because the distribution $\zeta$ over $[k]^N$ is not a product distribtution.  However, if we take a correlated equilibrium and condition $\zeta$ on the action of a random  agent $a_i$, then the total entropy of $\zeta|_{a_i}$ decreases which brings it closer to being a product distribution.  It stands to reason that conditioning on the actions of a few agents $S = \{a_{i_1}, a_{i_2},..., a_{i_y}\}$ can remove most of the correlations in $\zeta|_{S}$ thus bringing it to an approximate product distribution.  If $\zeta$ is a correlated equilibrium that continues to be a correlated equilibrium after conditioning on the actions of a small number of agents, then the resulting conditional distribution $\zeta|_{S}$ would be close to a mixed nash equilibrium.   One caveat is that the agents who's actions have been conditioned upon don't necessarily participate in the resulting Nash equilibrium so some small modifications will have to be made to the resulting equilibrium.  

Our goal is then to reformulate this high level idea into a convex program and design a rounding algorithm.  This is the subject of the next few sections, which we start with an overview of the fundamentals of polynomial optimization.  For a more thorough overview, see \cite{Nesterov00} \cite{Lasserre06} \cite{Shor87} \cite{Parrilo03}.  

\paragraph{Polynomial Optimization Fundamentals:}
First we define the pseudoexpectation functional $\pE[\cdot]$.  

\begin{definition}
Let $\bold{a} \defeq \{a_{ir}\}_{i \in [N], r \in [k]}$.  
A degree $\ell$ pseudoexpectation $\pE: \R[\bold{a}]^{\leq \ell} \rightarrow \R$ is a linear functional from the space of degree up to $\ell$ polynomials in variables $\bold{a}$ to the reals. Pseudoexpectations satisfy the following desirable properties. 
\begin{enumerate}
    \item Firstly, $\pE[1] = 1$ just like an expectation.
    \item Secondly, pseudoexpectations are linear functionals so that for any collection of degree less than $\ell$ polynomials $\{P_i(\bold{a})\}_{i \in [t]}$, we have  $\pE[\sum_{i=1}^t P_i(\bold{a})] = \sum_{i=1}^t \pE[P_i(\bold{a})]$. 
    \item Finally, for sum of squares pseudoexpectation functionals, we additionally require that $\pE[SoS_{\ell}(\bold{a})] \geq 0$ where $SoS_{\ell}(\bold{a})$ is the space of polynomials of degree less than $\ell$ that can be written as the sum of squares of polynomials.  
\end{enumerate}
\end{definition}

We are often searching for a pseudoexpectation functional satisfying certain constraints.  To distinguish between different pseudoexpectation functionals, we index $\pE_{\zeta}$ by an index $\zeta$.  Occasionally, we will refer to $\zeta$ as a pseudodistribution, but this is only for indexing purposes. Pseudoexpectations for degree $\ell$ polynomials over $N$ variables can be found in $N^{O(\ell)}$ time with some additional care required for bit precision issues.  This follows from the duality of the sum of squares proof system and semidefinite programming see \cite{Nesterov00} \cite{Lasserre06} \cite{Shor87} \cite{Parrilo03}.      

\begin{algorithm}
\caption{Convex Hierarchy for Polymatrix Nash Equilibrium}\label{alg:convex}
\textbf{Input: } Polymatrix game $(G,\{f_{ij}\}_{i,j \in [N]})$ \\
Let $\ell = \frac{L^4\ln(k)}{\epsilon^4}$, and search for a degree $\ell$ pseudoexpectation $\pE_\zeta[\cdot]$ satisfying the following  
\begin{enumerate}
    \item \textbf{Conditioning Preserving Correlated Equilibrium:} For all $i \in [N]$, for all $w \in [k]$, and for all $S_i \subset \{\{a_{jr}\}_{j \in [N]/i}\}_{r \in [k]}$ of size $|S_i| \leq \ell-2$.  
    \[\pE_\zeta[(\sum_{j \in N(i)} \sum_{p,q \in [k]^2}f_{ij}(p,q)a_{ip}a_{jq} - \sum_{j \in N(i)} \sum _{q \in [k]}f_{ij}(w,q)a_{jq} +\epsilon)\prod_{a_{uv} \in S_i}a_{uv}] \geq 0\] 
    \item \textbf{Booleanity:} $\forall i \in [N]$, $\forall p \in [k]$, and  for all $S \subset \{a_{ir}\}_{i \in [N], r \in [k]}$ satisfying $|S| \leq \ell - 2$ we have
    \[\pE_\zeta[(a_{ip}^2 - a_{ip})\prod_{a_{uv} \in S}a_{uv}] = 0\]
    \item \textbf{Coloring:} $\forall i \in [N]$ and  for all $S \subset \{a_{ir}\}_{i \in [N], r \in [k]}$ of size $|S| \leq \ell - 1$ we have  
    \[\pE_\zeta[(\sum_{p=1}^k a_{ip} - 1)\prod_{ a_{uv} \in S}a_{uv}] = 0\]
    \item \textbf{Normalization:} $$\pE_\zeta[1] = 1$$
    \item \textbf{Sum of Squares:} $\forall p(\bold{a}) \in SoS_{\ell}(\{a_{ir}\}_{i \in [N], r \in [k]})$ we constrain  
    \[\pE_\zeta[p(\bold{a})] \geq 0\] 
\end{enumerate}
\end{algorithm}

\paragraph{Formulating a Semidefinite Program: }
\pref{alg:convex} searches for a pseudoexpectation $\pE_\zeta$ satisfying the following constraints.  We have the normalization and sum of squares constraints, which are standard.  Furthermore we have the booleanity and coloring constraints which are also standard.  Collectively the booleanity and coloring constraints enforce that $a_{ip} \in \{0,1\}$ and that each player can adopt only one action at a time in expectation.  These constraints are also relaxed through the sherali adams hierarchy.  This implies that the constraints still hold for any set of fixings of the actions of up to $\ell-2$ players.  Finally, we have the conditioning preserving correlated equilibrium constraint.  For all $i \in [N]$, for all $w \in [k]$, and for all $S_i \subset \{\{a_{jr}\}_{j \in [N]/i}\}_{r \in [k]}$ of size $|S_i| \leq \ell-2$.  
    \[\pE_\zeta[(\sum_{j \in N(i)} \sum_{p,q \in [k]^2}f_{ij}(p,q)a_{ip}a_{jq} - \sum_{j \in N(i)} \sum _{q \in [k]}f_{ij}(w,q)a_{jq} +\epsilon)\prod_{a_{uv} \in S_i}a_{uv}] \geq 0\]           
This constraint ensures that the equilibrium (possibly correlated) we are searching for continues to be an equilibrium even if any small subset of $\ell-2$ players fix their actions.  Here the equilibrium condition is constrained to hold for every player who's actions are not fixed.  Ultimately, the conditioning preserving correlated equilibrium constraint is our replacement for both the Nash equilibrium constraint of $\calQ$ and the product distribution constraint   $\E_\zeta[a_{ip}a_{jq}] = \E_\zeta[a_{ip}]\E_\zeta[a_{jq}]$ of $\calQ$.  It turns out that the conditioning preserving correlated equilibria are close to mixed nash equilibria in a very precise sense which we'll exploit in the rounding algorithms detailed in \pref{sec:rounding}.          

\paragraph{Hybridized Hierarchies and Proof Systems: }
Note that we are neither working with the Sherali-Adams nor Sum of Squares relaxation.  We are also not working with a standard hybridized hierarchy because the Nash constraint for player $i$ only holds for multilinear polynomials that do not include $\{a_{ir}\}_{r \in [k]}$. As such, it will be important to design a rounding algorithm which is oblivious to the weakenings in the relaxation. 

Although it is possible that the sum of squares relaxation of pure nash equilibrium is feasible, it's not immediately obvious how to prove such a statement as the pure nash equilibrium is not guaranteed to exist.  We circumvent this problem, by explicitly designing the relaxation hierarchy for the polymatrix mixed nash equilibrium.  Note, that the relaxation is infeasible without the L-smooth payoffs assumption.

\paragraph{Explicit Constraints:} As we're not working with a standard relaxation, it's important to explicitly list the constraints of our convex program. These details are left for \pref{sec:explicit} 

\paragraph{Runtime: } The runtime of our algorithm is dominated by the convex program solve of \pref{alg:convex}.  The subsequent rounding analysis succeeds with small additive error in the constraints, which allows us to apply the ellipsoid algorithm \cite{GLS81}.  The subsequent rounding procedures require no more than $(Nk)^{O(\frac{L^4 \ln(k)}{\epsilon^4})}$ time upper bounded by an exhaustive search over the variable space. 

\begin{lemma} \label{lem:runtime}
\pref{alg:convex} can be solved via the ellipsoid algorithm \cite{GLS81} to $2^{-\poly(N,k)}$ additive accuracy in $(Nk)^{O(\frac{L^4 \ln(k)}{\epsilon^4})}$ time.      
\end{lemma}

\begin{proof}
\pref{alg:convex} is a convex program in $(Nk)^{O(\frac{L^4\ln(k)}{\epsilon^4})}$ variables and $(Nk)^{O(\frac{L^4\ln(k)}{\epsilon^4})}$ convex and explicitly bounded constraints with a feasible region of size greater than $2^{-\poly(N,k)}$.  Furthermore, both \pref{lem:rounding-main} and \pref{alg:bestrespond} hold for any $2^{-\poly(N,k)}$ additive error in the constraints. Therefore, the ellipsoid algorithm \cite{GLS81} produces a solution to \pref{alg:convex} in $(Nk)^{O(\frac{L^4 \ln(k)}{\epsilon^4})}$ time.         
\end{proof}

\section{Rounding Algorithms} \label{sec:rounding}
In this section we present our rounding algorithm.  Our rounding proceeds in two phases.  In the first phase we adopt the conditioning procedure of \cite{BRS11} to bring $\pE_{\zeta}$ close to a product distribution which we denote $\pE_{\hat{\zeta}}$.  Here, to achieve quasipolynomial running time we must improve the $k$ dependence of the \cite{BRS11} analysis. Towards these ends, we develop a novel polynomial approximation argument of the pseudo-covariance matrix, which accelerates the \cite{BRS11} rounding at the expense of an additional $\poly(k)$ in the edge density of the graph.  This argument is the key to our quasipolynomial time algorithm.   

In the second phase of our rounding, we must take the approximate product distribution $\pE_{\hat{\zeta}}$ and turn it into a mixed nash equilibrium.  This is done via the  BestResponsePursuit \pref{alg:brp}.  At a high level, \pref{alg:brp} takes all players which do not have the mixed nash equilibrium constraint satisfied, which we denote $K_1$, and attempts to set them to best respond to the mixed strategies of the players in $V/K_1$.  Of course this does not work because $K_1$ could be a clique.  This would only work, if the following condition is satisfied.  For all $i \in V$, the fraction of edges between $i$ and $K_1$ is small i.e $\frac{N_{K_1}(i)}{d} \leq \frac{\epsilon}{L}$.  If this is the case, then setting every player in $K_1$ to best respond to the mixed strategies of player in $V/K_1$ would result in an $2\epsilon$ approximate mixed nash equilibrium.  This is a consequence of the $L$-smoothness assumption.  

A priori, $K_1$ does not satisfy this condition.  However, given $K_1$ it is possible to show that there exists a set $W$ such that $K_1 \subset W$, $|W| \leq \frac{4}{3}|K_1|$, and for all $i \in V/W$, we have $\frac{N_W(i)}{d} \leq \frac{\epsilon}{L}$. Then we can hope to form a set $U \subseteq W$ such that $\frac{|U|}{|W|} \geq \frac{9}{10}$ and for all $i \in U$, $\frac{N_W(i)}{d} \leq \frac{\epsilon}{L}$.  This allows us to set the players in $U$ to best respond to the mixed strategies of players in $V/W$.  After this is done, we are left with a set of players $K_2 = W/U$, where $|K_2| \leq \frac{|K_1|}{5}$.  Now iterating the entire procedure on $K_2$ instead of $K_1$ we can hope to iteratively contract the set of players not participating in the mixed nash equilibrium until only a small set of players are left.  This small set of players can then all be set to best respond.

Each iteration of the above argument comes at a cost to the approximation of the mixed nash equilibrium so it is important to keep track of the sources of error as they add up over multiple iterations.  The details of this argument is the subject of \pref{sec:best-response-pursuit}. Next we state the guarantees of the conditioning rounding procedure.  To start we define a notion of distance quantifying how far a probability distribution is from a product distribution.  

\begin{definition}
Let $\Delta(i)$ be the local correlation at node $i \in [N]$ defined to be.   
\[\Delta(i) \defeq \frac{1}{|N(i)|}\sum_{j \in N(i)} \sum_{(p,q) \in [k]^2} |\pE[a_{ip} a_{jq}] - \pE[a_{ip}]\pE[a_{jq}]|]\]
Let $\Delta$ be the average local correlation of the nodes in $G$.  
\[\Delta \defeq \frac{1}{2|E|}\sum_{(i,j) \in E} \sum_{(p,q) \in [k]^2} |\pE[a_{ip} a_{jq}] - \pE[a_{ip}]\pE[a_{jq}]|]\]
\end{definition}

Intuitively, correlation/covariance capture how far a pair of random variables $a_i$ and $a_j$ are from a product distribution.  The average covariance over all pairs of random variables over all edges of $G$ is then a measure of how far a distribution in $\calP(k,V)$ is from a product distribution.

\begin{fact} [Conditioning] For a degree $\ell$ pseudoexpectation $\pE_{\zeta}$ over the polynomial system $\{x_i^2 = x_i\}_{i=1}^N$, the following conditioning procedure is well defined.  

\[\pE_{\zeta}[p(x_1,...,x_N) | x_i = 1] \defeq \frac{\pE_{\zeta}[p(x_1,...,x_N)x_i]}{\pE_{\zeta}[x_i]}\]
for any $p(x_1,...,x_N) \in \R[x_1,...,x_N]^{\leq \ell}$
\end{fact}

Note that the conditioning preserving correlated equilibrium constraint for player $i$ is preserved after conditioning on the actions of any player $j \neq i$.  Furthermore, this holds for any set of $\ell-1$ such conditionings.  Next, we will prove our conditioning rounding result for the class of low threshold rank graphs, which are a generalization of expander graphs and include the $G(N,p)$ graphs we study in this work.    

\begin{definition}
We say a graph $G$ has $(\rho, \tau)$ threshold rank if 
the normalized adjacency matrix $\frac{1}{d}A$ of $G$ with eigenvalues $\{\lambda_i\}_{i=1}^N$ satisfies $|\{i \in [N]| \lambda_i \geq \rho\}| = \tau$. 
\end{definition}
In particular, the $G(N,p)$ random graphs considered in this work have $(\frac{1}{\sqrt{d}},1)$ threshold rank. The following lemma is the main rounding guarantee of the conditioning procedure.    
\begin{lemma} \label{lem:rounding-main}
Consider the following procedure, 
\begin{enumerate}
    \item Initialize seed set $S = \emptyset$, and assignments $\cal{R} = \emptyset$
    \item Let $a_1,a_2,...,a_N$ be random variables such that $\mathbb{P}(a_i = p) = \pE_\zeta[a_{ip}]$ for all $p \in [k]$ and for all $i \in [N]$.
    \item Let $i \in V/S$ be drawn uniformly at random, and append $i$ to $S$
    \item draw $p_i \sim a_i$ and append $p_i$ to $\cal{R}$
    \item Update $\pE_\zeta[\cdot] \gets \pE_\zeta[\cdot|a_i = p_i]$
    \item While $|S| \leq \frac{L^4\ln(k)}{\epsilon^4}-2$ go to step (2.)
\end{enumerate}
Let $\pE_{\hat{\zeta}}$ be the pseudoexpectation that is the output of the algorithm above, and let $\Delta_{\hat{\zeta}}$ be its associated local correlation.  When the algorithm terminates we have $S \defeq \{i_1,i_2,...,i_{r}\}$ and $R \defeq \{p_{i_1}, p_{i_2}, ..., p_{i_r}\}$. Let $\cal{U}$ be the distribution over $\{(i_y,p_y)\}_{y \in [r]}$ induced by the algorithm.  Then 
$$\E_{\{(i_y,p_y)\}_{y \in [r].  } \sim \cal{U}}[\Delta_{\hat{\zeta}}|a_{i_1} = p_{i_1}, a_{i_2} = p_{i_2} ,...,a_{i_r} = p_{i_r}] \leq \frac{\epsilon^2}{L^2}$$ 
\end{lemma}

The conditioning algorithm selects a player at random and fixes/conditions their actions according to the local probability distribution prescribed by $\pE_{\zeta}$.  Iterating this algorithm reduces the local correlation.  The proof of \pref{lem:rounding-main} is left for \pref{sec:mutual-information}.  Additionally, it turns out there's no need for the procedure to be executed iteratively.  One could simply enumerate all possible conditionings of the actions of $\ell -2$ agents.  One such conditioning is guaranteed to bring the local correlation down to $\epsilon$.  Next, we move onto the guarantees of \pref{alg:brp}.  

\begin{restatable}[Best Response Pursuit]{lem}{bestresponsepursuit}
\label{lem:best-response-pursuit}
Let $\pE_{\zeta}$ be a pseudoexpectation satisfying the constraints in \pref{alg:convex}. Let $\Delta_\zeta \leq \epsilon^2$ be the local correlation of $\pE_{\zeta}$.  Let $K_1 \subset V$ be the set of nodes satisfying $\Delta_\zeta(i) \geq 2\epsilon$. Let $\mu_{ip} = \pE[a_{ip}] \text{   for all }  i \in [N] \text{ and for all } p \in [k]$.  Let $\Gamma_1 \gets \{\{\mu_{ip}\}_{p=1}^k\}_{i=1}^N$.  
We define $K_2,K_3,...,K_T$ and $\Gamma_2,\Gamma_3,...,\Gamma_T$ iteratively as follows. Let $W = Flag(K_t,G)$ and let $(K_{t+1}, \Gamma_{t+1}) = BestRespond(W,\Gamma_t)$.  Here we iterate until $K_T = \emptyset$.  Then we have for all $i \in V$ with probability $1 - N^{-\ln^2(N)}$ over the randomness of the graph,    
\[\sum_{j \in N(i)} \sum_{(p,q) \in [k]^2} f_{ij}(p,q) \pE_{\Gamma_T}[a_{ip}]\pE_{\Gamma_T}[a_{jq}] - \sum_{q \in [k]} f_{ij}(w,q) \pE_{\Gamma_T}[a_{jq}]  \leq 120L\epsilon\] 
\end{restatable} 
Given \pref{lem:best-response-pursuit} we can conclude the following guarantee holds for our rounding algorithm.  
\begin{lemma} [Rounding] \label{lem:best-response-main}
Let $\pE_{\zeta}$ be the pseudoexpectation output by \pref{alg:convex}, which is then passed to the conditioning procedure in \pref{lem:rounding-main}.  Let $\pE_{\hat{\zeta}}$ be the pseudoexpectation output by the conditioning procedure in \pref{lem:rounding-main}. Then \pref{alg:brp}, BestResponsePursuit($\pE_{\hat{\zeta}}$) outputs a product distribution $\Gamma \defeq \{\mu_{ip}\}_{i \in [N], p \in [k]} \in \calP(k,V)$ satisfying that for all $i \in [N]$ and for all $w \in [k]$

\[\sum_{j \in N(i)} \sum_{p,q \in [k]^2}f_{ij}(p,q)\E_\Gamma[a_{ip}]\E_\Gamma[a_{jq}] - \sum_{j \in N(i)} \sum_{q \in [k]}f_{ij}(w,q)\E_\Gamma[a_{jq}] \geq -\epsilon\]
That is to say $\Gamma$ is an $\epsilon$-approximate mixed nash equilibrium
\end{lemma}

\begin{proof}
Using \pref{lem:rounding-main} we obtain $\pE_{\hat{\zeta}}$ satisfying local correlation $\Delta \leq \frac{\epsilon^2}{L^2}$.  Then using \pref{lem:best-response-pursuit} we obtain a product distribution $\Gamma$ satisfying  

\[\sum_{j \in N(i)} \sum_{p,q \in [k]^2}f_{ij}(p,q)\E_\Gamma[a_{ip}]\E_\Gamma[a_{jq}] - \sum_{j \in N(i)} \sum_{q \in [k]}f_{ij}(w,q)\E_\Gamma[a_{jq}] \geq -\epsilon\]
as desired.  
\end{proof}
We now have all the ingredients to prove \pref{thm:main-theorem}
\maintheorem*

\begin{proof}
We have by \pref{lem:runtime} that \pref{alg:convex} can be solved in $(Nk)^{O(\frac{L^4 \ln(k)}{\epsilon^4})}$ time.  We have by \pref{lem:best-response-main} that the conditioning procedure and \pref{alg:brp} produce an $\epsilon$-approximate Nash equilibrium with high probability $1 - N^{-\ln^2(N)}$.  The runtime of both procedures is upper bounded by the time it takes to exhaustively search over all $\ell-2$ size subsets of conditionings of $\bold{a}$.  This can be done in time $(Nk)^{O(\frac{L^4 \ln(k)}{\epsilon^4})}$.  
\end{proof}

\begin{algorithm}
\caption{Best Response Pursuit}\label{alg:brp}
\textbf{Input: } $\pE_{\zeta}$\;
\KwResult{$\Gamma = \{\{\mu_{ip}\}_{p=1}^k\}_{i=1}^N$}
$\mu_{ip} \gets \pE[a_{ip}] \text{   for all }  i \in [N] \text{ and for all } p \in [k]$\;
$K_1 \gets \emptyset$\;
\For{$i \in [N]$}
{
\If{
$\Delta(i) \geq 2\epsilon$}{$K_1 \gets K_1 \cup \{i\}$}
}
$\Gamma_1 \gets \{\{\mu_{ip}\}_{p=1}^k\}_{i \in V/K_1}$ \;
\For{$t \in [T]$}{
$W = Flag(K_t,G)$\;
$(\Gamma_{t+1}, K_{t+1}) \gets BestRespond(W, \Gamma_t)$\;
$\Gamma \gets \Gamma_{t+1}$\;
\If{$K_{t+1} = \emptyset$}{ \textbf{break}} 
}

$\textbf{Return: } \Gamma$
\end{algorithm}

\begin{algorithm}
\caption{Flag}\label{alg:flag}
\KwData{$K \subseteq V$}
\KwResult{$W \subseteq V$}
$W \gets K$\;
$H_{old} \gets K$\;
\While{$|H_{old}| \geq \ln^3(N)$}{
$H_{new} \gets \emptyset$\;
\For{$i \in \overline{W}$}{
\If{$\frac{|N_{H_{old}}(i)|}{d} \geq \max(10 \frac{|H_{old}|}{N}, 10\frac{\ln(n)}{d})$}{
$H_{new} \gets H_{new} \cup \{i\}$\;
}
}
$W \gets W \cup H_{new}$\;
$H_{old} \gets H_{new}$
}

$\textbf{Return: } W$

\end{algorithm}

\begin{algorithm}
\caption{BestResponse($\Gamma,W$)}\label{alg:bestrespond}
\textbf{Input: }$\Gamma = \{\{\mu_{ip}\}_{p=1}^k\}_{i \in V/W}, W$\;
\KwResult{$\hat{\Gamma} = \{\{\hat{\mu}_{ip}\}_{p=1}^k\}_{i \in V/Q}, Q$}
$\{\hat{\mu}_{ip}\}_{p=1}^k\}_{i \in V/W} \gets  \{\{\mu_{ip}\}_{p=1}^k\}_{i \in V/W}$ \;
\If{$|W| \leq \ln^3(N)$}{
$Q \gets \emptyset$
} \Else{
$Q \gets \emptyset$\;
\For{$i \in W$}{
\If{$\frac{|N_{W}(i)|}{|N(i)|} \geq \epsilon$}{
$Q \gets \{i\}$\;
} 
}
}

\For{$i \in W/Q$}{
\For{$p \in [k]$}{
\If{$\E_{\Gamma}[\sum_{j \in N_{\overline{W}}(i)}\sum_{q\in [k]} f_{i,j}(p,q)a_{jq}] \geq \max_{p' \in [k]} \E[\sum_{j \in N_{\overline{W}}(i)}\sum_{q\in [k]} f_{i,j}(p',q)a_{jq}]$}{
$r_{ip} \gets 1$\;
} \Else{
$r_{ip} \gets 0$
}
}
}
$\hat{\mu}_{ip} \gets \frac{r_{ip}}{\sum_{p \in [k]}r_{ip}} $ \text{ for all } $i \in W/Q$ \text{ and for all } $p \in [k]$\\
$\textbf{return: } (\{\{\hat{\mu}_{ip}\}_{p=1}^k\}_{i \in V/Q},Q)$\;
\end{algorithm}

\section{Propagation of Mutual Information}\label{sec:mutual-information}
In the next two sections we work towards a proof of \pref{lem:rounding-main}.  First we present an upper bound on the local correlation $\Delta_{\zeta}$ in terms of the global mutual information defined to be $\frac{1}{N^2}\sum_{i,j \in [N]^2}I(a_i,a_j)$.  Here $\{a_i\}_{i \in [N]}$ are the $\ell$-local random variables over the agent's actions provided by the level $\ell$ sherali-adams hierarchy on the booleanity constraints.  The pairwise mutual information $I(a_i,a_j)$ is therefore well defined.  Once we relate $\Delta_{\zeta}$ to  $\frac{1}{N^2}\sum_{i,j \in [N]^2}I(a_i,a_j)$ we can use a standard conditioning argument to show that after a small number of rounds of conditioning $\frac{1}{N^2}\sum_{i,j \in [N]^2}I(a_i,a_j)$ is small.  This would in turn show that $\Delta_{\zeta}$ is small.   

\begin{lemma} \label{lem:loc_to_glob}
Let $\Delta_{\zeta}$ denote the local correlation of pseudodistribution $\pE_\zeta$ on a $(\rho,\tau)$ threshold rank graph $G$ defined as
\[\Delta_{\zeta} \defeq \frac{1}{2|E|}\sum_{(i,j) \in E} \sum_{(p,q) \in [k]^2} |\pE[a_{ip} a_{jq}] - \pE[a_{ip}]\pE[a_{jq}]|]\]
Let $f(x) = \sum_{r=1}^p \alpha_r x^r$ be a univariate polynomial satisfying $|x| \leq f(x) \leq |x| + \epsilon'$ for $x \in [-1,1]$.  Let $H \in \R^{kN \times kN}$ be a matrix row indexed by $i \in [N]$ and $p \in [k]$ and column indexed by $j \in [N]$ and $q \in [k]$.  We define $H_{\{(i,p),(j,q)\}} = f(\pE[a_{ip}a_{jq}] - \pE[a_{ip}]\pE[a_{jq}])$.  
Then we have that 
\[\Delta_{\zeta}  \leq
\frac{\rho k}{N} \|H\|_* + \tau \sqrt{2(\epsilon'k^2)^2 + \frac{2}{N^2}\sum_{i,j \in [N]^2}I(a_i,a_j)} \]
\end{lemma}

\begin{proof}
We begin with the local correlation 
\[\Delta_\zeta \defeq \frac{1}{2|E|}\sum_{(i,j) \in E} \sum_{(p,q) \in [k]^2} |\pE[a_{ip}a_{jq}] - \pE[a_{ip}]\pE[a_{jq}]|] = \frac{1}{|E|}\langle A, K\rangle  \]
where $K_{ij} := \sum_{(p,q) \in [k]^2} |\pE[a_{ip} a_{jq}] - \pE[a_{ip}]\pE[a_{jq}]|$.  Let $M$ be a matrix such that $M_{ij} = \sum_{(p,q) \in [k]^2}f(\pE[a_{ip}a_{jq}] - \pE[a_{ip}]\pE[a_{jq}])$ for univariate polynomial $f$ satisfying $|x| \leq f(x) \leq |x| + \epsilon'$.  Here $f(x)$ is the polynomial defined in \pref{lem:jackson}.  We use the univariate polynomial entrywise upper bound to obtain 
\[\frac{1}{|E|}\langle A, K \rangle  \leq \frac{1}{|E|}\langle A, M\rangle\]
Now we split the eigenvalues of $G$ into those larger than $\rho$ and those less than $\rho$.  

\[ \frac{1}{|E|}\langle A, M \rangle  = \frac{1}{N}\langle \frac{N}{|E|}A, M \rangle = \frac{1}{N}\sum_{i=1}^N \lambda_i \langle v_iv_i^T, M \rangle = \frac{1}{N}(\sum_{|\lambda_i| \leq \rho} \lambda_i \langle v_iv_i^T, M \rangle + \sum_{|\lambda_i| \geq \rho} \lambda_i \langle v_iv_i^T, M \rangle) \]
Taking absolute values in the second term we upper bound by 
\[ \leq \frac{1}{N}\Big(\sum_{|\lambda_i| \leq \rho} \lambda_i \langle v_iv_i^T, M \rangle + \sum_{|\lambda_i| \geq \rho} |\lambda_i| |\langle v_iv_i^T, M \rangle|\Big) \]
Using the fact that the eigenvalues of the normalized adjacency matrix is upper bounded by $1$ and taking cauchy schwarz we obtain 
\[ \leq \frac{1}{N}\Big(\sum_{|\lambda_i| \leq \rho} \lambda_i \langle v_iv_i^T, M \rangle + \sum_{|\lambda_i| \geq \rho} (\sum_{(r,s) \in [N]^2} v_{ir}^2v_{is}^2)^{1/2} \| M\|_F\Big) \]

\[ = \frac{1}{N}\Big(\sum_{|\lambda_i| \leq \rho} \lambda_i \langle v_iv_i^T, M \rangle + \sum_{|\lambda_i| \geq \rho} (\sum_{r=1}^{N} v_{ir}^2) \| M\|_F\Big) \]

Using the fact that $\|v_i\|^2 = 1$ we obtain 
\[ = \frac{1}{N}\Big(\sum_{|\lambda_i| \leq \rho} \lambda_i \langle v_iv_i^T, M \rangle + \sum_{|\lambda_i| \geq \rho} \| M\|_F\Big) \]
Using the $(\rho,\tau)$ threshold rank condition we obtain.  
\[ = \frac{1}{N}\Big(\sum_{|\lambda_i| \leq \rho} \lambda_i \langle v_iv_i^T, M \rangle + \tau \sqrt{\sum_{i,j \in [N]^2}(M_{ij})^2}\Big) \]
Let $v^{\oplus k} = (v,v,...,v) \in \R^k$.  Let $\overline{v}_i := (v_{i1}^{\oplus k}, v_{i2}^{\oplus k}, ... , v_{iN}^{\oplus k}) \in \R^{kN}$.  Let $H \in \R^{kN \times kN}$ be a matrix row indexed by $i \in [N]$ and $p \in [k]$ and column indexed by $j \in [N]$ and $q \in [k]$.  We define $H_{\{(i,p),(j,q)\}} = f(\pE[a_{ip}a_{jq}] - \pE[a_{ip}]\pE[a_{jq}])$.  Let $\phi_1,...,\phi_{kN}$ be the eigenvalues of $H$ with corresponding eigenvectors $q_1,...,q_{kN}$.  Then we have 

\[=  \frac{1}{N}\Big(\sum_{|\lambda_i| \leq \rho} \lambda_i \langle \overline{v}_i\overline{v}_i^T, H \rangle + \tau \sqrt{\sum_{i,j \in [N]^2}(M_{ij})^2}\Big) \]

\[=  \frac{1}{N}\Big(\sum_{|\lambda_i| \leq \rho} \lambda_i \langle \overline{v}_i\overline{v}_i^T, \sum_{j=1}^{kN} \phi_j q_jq_j^T \rangle + \tau \sqrt{\sum_{i,j \in [N]^2}(M_{ij})^2}\Big) \]

\[=  \frac{1}{N}\Big(\sum_{j=1}^{kN} \phi_j \sum_{|\lambda_i| \leq \rho} \lambda_i \langle \overline{v}_i\overline{v}_i^T, q_jq_j^T \rangle + \tau \sqrt{\sum_{i,j \in [N]^2}(M_{ij})^2}\Big) \]
Applying absolute values and noting that $\langle \overline{v}_i\overline{v}_i^T, q_jq_j^T \rangle$ is a square we upper bound by 
\[\leq  \frac{1}{N}\Big(\sum_{j=1}^{kN} |\phi_j| \sum_{|\lambda_i| \leq \rho} |\lambda_i| \langle \overline{v}_i\overline{v}_i^T, q_jq_j^T \rangle + \tau \sqrt{\sum_{i,j \in [N]^2}(M_{ij})^2}\Big) \]
Upper bounding $\lambda_i \leq \rho$ we obtain 
\[\leq  \frac{1}{N}\Big(\sum_{j=1}^{kN} |\phi_j| \sum_{|\lambda_i| \leq \rho} \rho k \langle \frac{1}{\sqrt{k}}\overline{v}_i \frac{1}{\sqrt{k}}\overline{v}_i^T, q_jq_j^T \rangle + \tau \sqrt{\sum_{i,j \in [N]^2}(M_{ij})^2}\Big) \]

Using the fact that $\langle \frac{1}{\sqrt{k}}\overline{v}_i \frac{1}{\sqrt{k}}\overline{v}_i^T, q_jq_j^T \rangle$ is a square for all $i \in [N]$ we upper bound by 

\[\leq  \frac{1}{N}\sum_{j=1}^{kN} |\phi_j|  \rho k\sum_{i \in [N]} \langle \frac{1}{\sqrt{k}}\overline{v}_i \frac{1}{\sqrt{k}}\overline{v}_i^T, q_jq_j^T \rangle + \frac{\tau}{N} \sqrt{\sum_{i,j \in [N]^2}(M_{ij})^2} \]

\[= \frac{1}{N}\sum_{j=1}^{kN} |\phi_j|  \rho k\sum_{i \in [N]} \langle \frac{1}{\sqrt{k}} \overline{v_i},q_j\rangle^2 + \frac{\tau}{N} \sqrt{\sum_{i,j \in [N]^2}(M_{ij})^2} \]

We know that $\frac{1}{\sqrt{k}}\overline{v}_i$ for $i \in [N]$ forms an orthonormal basis, and $\|q_j\|^2 = 1$.  Therefore, we upper bound the first term by   

\[\leq \frac{1}{N}\sum_{j=1}^{kN} |\phi_j|  \rho k + \frac{\tau}{N} \sqrt{\sum_{i,j \in [N]^2}(M_{ij})^2} \]

\[=  \frac{1}{N}\sum_{j=1}^{kN} |\phi_j|\rho k + \frac{\tau}{N} \sqrt{\sum_{i,j \in [N]^2}(M_{ij} - \sum_{(p,q) \in [k]^2}|cov(a_{ip}, a_{jq})| + \sum_{(p,q) \in [k]^2}|cov(a_{ip}, a_{jq})|)^2}\]

Applying $(a+b)^2 \leq 2a^2 + 2b^2$ to the second term we obtain

\[\leq  \frac{1}{N}\sum_{j=1}^{kN} |\phi_j|\rho k + \tau \sqrt{\frac{1}{N^2}\sum_{i,j \in [N]^2}\big(2(M_{ij} - \sum_{(p,q) \in [k]^2}|cov(a_{ip}, a_{jq})|)^2 + 2(\sum_{(p,q) \in [k]^2}|cov(a_{ip}, a_{jq})|)^2}) \]

Using the fact that $M$ is entrywise $\epsilon'k^2$ close to $K$ we have  

\[\leq  \frac{\rho k}{N} \|H\|_* + \tau \sqrt{2(\epsilon'k^2)^2 + \frac{2}{N^2}\sum_{i,j \in [N]^2}(\sum_{(a,b) \in [k]^2}|cov(a_{ip}, a_{jq})|)^2} \]
Applying pinsker's inequality we obtain 
\[\leq  \frac{\rho k}{N} \|H\|_* + \tau \sqrt{2(\epsilon'k^2)^2 + \frac{4}{N^2}\sum_{i,j \in [N]^2}I(a_i,a_j)} \]
which is the desired inequality.  
\end{proof}

Now we have upper bounded $\Delta_{\zeta}$ by a relation involving $\frac{4}{N^2}\sum_{i,j \in [N]^2}I(a_i,a_j)$ and $\|H\|_*$.  In the next section we bound $\|H\|_*$.  
\section{Accelerated Rounding via Polynomial Approximation}
The following lemma is a self contained fact about the nuclear norm of matrices that are written as the sum of PSD matrices.    

\begin{lemma} \label{lem:nuclear_norm}
Let $f(x) = \sum_{r=1}^p \alpha_r x^r$.  Let $H \in \R^{kN \times kN}$ be a matrix row indexed by $i \in [N]$ and $p \in [k]$ and column indexed by $j \in [N]$ and $q \in [k]$.  We define $H_{\{(i,p),(j,q)\}} = f(\pE[a_{ip}a_{jq}] - \pE[a_{ip}]\pE[a_{jq}])$.  Then 
\[\|H\|_* \leq Nk\sum_{r=1}^p |\alpha_r|\]
In particular, for $f(x)$ defined in \pref{lem:jackson} where $|x| \leq f(x) \leq |x| + \epsilon'$
\[\|H\|_* \leq \frac{18Nk}{\sqrt{\pi}\epsilon'}\]
\end{lemma}

\begin{proof}
Let $Q_{\{(ip),(jq)\}} = cov(a_{ip},a_{jq}) = \langle v_{ip}, v_{jq} \rangle$.  Let $Q^{\odot r}$ be the entrywise $r$'th power of $Q$.  We know 

\[Q^{\odot r}_{\{(ip),(jq)\}} \defeq cov(a_{ip},a_{jq})^r = \langle v_{ip}, v_{jq} \rangle^r =  \langle v_{ip}^{\otimes r}, v_{jq}^{\otimes r}\rangle\]

Thus we have $Q^{\odot r}$ is positive semidefinite, which can also be seen by applying the Schur product theorem.  Thus we can upper bound the nuclear norm 
\[\|Q^{\odot r}\|_* = Tr(Q^{\odot r}) =  \sum_{i \in [N], p \in [k]} \|v_{ip}\|^{2r} \leq Nk\] 

where the inequality follows from the fact that $\langle v_{ip}, v_{jq} \rangle \leq 1$. This implies we can upper bound the nuclear norm of $H$ as follows 
 \[\| H \|_* = \| \sum_{r=1}^p \alpha_r Q^{\odot r} \|_* \leq \sum_{r=1}^p \|\alpha_r Q^{\odot r} \|_* = \sum_{r=1}^p |\alpha_r|\| Q^{\odot r} \|_* \leq Nk\sum_{r=1}^p |\alpha_r|\]
 Where the first inequality follows from triangle inequality, and the second inequality follows from the nonnegativity of the nuclear norm. Using \pref{lem:jackson} we upper bound by $\sum_{r=1}^p |\alpha_r| \leq \frac{18}{\sqrt{\pi}}p \leq \frac{18}{\sqrt{\pi}\epsilon'}$ which allows us to conclude
\[\| H \|_* \leq \frac{18Nk}{\sqrt{\pi}\epsilon'} \] 
as desired.  
\end{proof}

Finally we are ready to prove \pref{lem:rounding-main}
\begin{lemma} \label{lem:rounding-main}
Let $a_1,a_2,...,a_N$ be random variables such that $\mathbb{P}(a_i = p) = \pE[a_{ip}]$ for all $p \in [k]$ and for all $i \in [N]$.  Consider the following procedure, 
\begin{enumerate}
    \item Initialize seed set $S = \emptyset$, and assignments $\cal{R} = \emptyset$
    \item Let $i \in V/S$ be drawn uniformly at random, and append $i$ to $S$ 
    \item draw $p_i \sim a_i$ and append $p_i$ to $\cal{R}$
    \item While $|S| \leq \frac{\ln(k)}{L^4\epsilon^4}$ go to step (2.)
\end{enumerate}
When the procedure terminates we have $S \defeq \{i_1,i_2,...,i_{r}\}$ and $R \defeq \{p_{i_1}, p_{i_2}, ..., p_{i_r}\}$. Let $\cal{U}$ be the distribution over $\{(i_y,p_y)\}_{y \in [r]}$ induced by the procedure.  Then 
$$\E_{\{(i_y,p_y)\}_{y \in [r].  } \sim \cal{U}}[\Delta|a_{i_1} = p_{i_1}, a_{i_2} = p_{i_2} ,...,a_{i_r} = p_{i_r}] \leq \frac{\epsilon^2}{L^2}$$  

\end{lemma}
\begin{proof}
Let $\epsilon' = \frac{\epsilon^2}{k^2}$.  Then our bound on $\Delta$ is as follows. By \pref{lem:loc_to_glob} we have 
\[\Delta \leq  \frac{\rho k}{N} \|H\|_* + \tau \sqrt{2(\epsilon'k^2)^2 + \frac{2}{N^2}\sum_{i,j \in [N]^2}I(a_i,a_j)} \]
For $G(N,p)$ random graphs, with $d = pN$ we have $\rho = \frac{1}{\sqrt{d}}$ and $\tau = 1$ so we obtain 
\[=  \frac{ k}{N\sqrt{d}} \|H\|_* + \sqrt{2(\epsilon'k^2)^2 + \frac{2}{N^2}\sum_{i,j \in [N]^2}I(a_i,a_j)} \]
Setting $\epsilon' = \frac{\epsilon}{k^2}$ we obtain 
\[=  \frac{k}{\sqrt{d}N} \|H\|_* + \sqrt{2\epsilon^2 + \frac{2}{N^2}\sum_{i,j \in [N]^2}I(a_i,a_j)} \]
Applying \pref{lem:nuclear_norm} we have  $\|H\|_* \leq \frac{Nk}{\epsilon'}= \frac{Nk^3}{\epsilon}$ which yields 
\[\leq  \frac{ k}{N\sqrt{d}} N \frac{k^3}{\epsilon} + \sqrt{2\epsilon^4 + \frac{2}{N^2}\sum_{i,j \in [N]^2}I(a_i,a_j)} \]
Setting $d = \frac{k^8}{\epsilon^4}$ we obtain 
\[\leq  \epsilon + \sqrt{2\epsilon^2 +  \frac{2}{N^2}\sum_{i,j \in [N]^2}I(a_i,a_j)}\]

At this point we have established 
\begin{multline}
\E_{\{(i_y,p_y)\}_{y \in [r]  } \sim \cal{U}}[\Delta|a_{i_1} = p_{i_1}, a_{i_2} = p_{i_2} ,...,a_{i_r} = p_{i_r}]\\ \leq  \epsilon + \E_{\{(i_y,p_y)\}_{y \in [r]  } \sim \cal{U}}\left[\sqrt{2\epsilon^2 +  \frac{2}{N^2}\sum_{i,j \in [N]^2}I(a_i,a_j|a_{i_1}=p_{i_1}, a_{i_2}=p_{i_2}, ..., a_{i_r}=p_{i_r})}\right]    
\end{multline}

We have by Jensen's inequality 
\[\leq  \epsilon + \sqrt{2\epsilon^2 + \E_{\{(i_y,p_y)\}_{y \in [r]  } \sim \cal{U}}\left[ \frac{2}{N^2}\sum_{i,j \in [N]^2}I(a_i,a_j|a_{i_1}=p_{i_1}, a_{i_2}=p_{i_2}, ..., a_{i_r}=p_{i_r})\right]} \]

Finally via \pref{lem:global_correlation} we have $\frac{1}{N^2}\sum_{i,j \in [N]^2}I(a_i,a_j|a_{i_1}, a_{i_2}, ..., a_{i_r}) \leq \epsilon^2$ after $\frac{\ln(k)}{\epsilon^2}$ rounds of conditioning.  
Thus we conclude

$$\E_{\{(i_y,p_y)\}_{y \in [r]} \sim \cal{U}}[\Delta|a_{i_1} = p_{i_1}, a_{i_2} = p_{i_2} ,...,a_{i_r} = p_{i_r}] \leq \epsilon$$
as desired. If we used $\frac{L^4\ln(k)}{\epsilon^4}$  
\end{proof}

In the next section we present the guarantees of the Best Response Pursuit \pref{alg:brp}.  
\section{Best Response Pursuit}  \label{sec:best-response-pursuit}
When local correlation is smaller than $\epsilon^2$,  this implies via markov that the number of nodes $i$ with $\Delta(i) \geq \epsilon$ is less than $ \epsilon N$.  Let these 'bad' nodes be $K_1$, and let $\overline{K_1}$ be the set of 'good' nodes.  A naive strategy is to set all the nodes in $K_1$ to best respond to the nodes in $\overline{K}_1$.  This would succeed if all the nodes in $V$ had less than a $10\epsilon$ fraction of edges incident to $K_1$.  Setting all the nodes in $K_1$ to best respond would satisfy the Nash equilibrium constraints to error $10L\epsilon$.  Thus, a preliminary goal is to assemble a set $K$ such that $K_1 \subseteq W$ satisfying that every node in $\overline{W}$ satisfies that no more than a $10\epsilon$ fraction of edges are incident to $W$. Given such a set $W$, we could set all the nodes in $W$ to best response and the nodes $i \in \overline{W}$ will  continue to satisfy $\Phi(i) \leq \epsilon + 10L\epsilon$.  Unfortunately, this does not guarantee that all the nodes $i \in W$ will satisfy $\Phi(i) \leq 10L\epsilon$ as their may be nodes in $W$ with too many edges in $W$.  Indeed, we will not be able to set all the nodes in $W$ to best response.  However, we can set most of the nodes in $W$ to best respond. This leaves us with a set $K_2 \subset W$ which has not been set to best response with $|K_2| \leq \frac{|K_1|}{2}$.  Now we can iterate the same procedure on $K_2$ as had been performed on $K_1$ which contracts the set of bad nodes geometrically whilst controlling the error in the constraints of the good nodes.  

\subsection{Best response preliminaries: }
\paragraph{Notation:} 
let $Q \subseteq V$ and let $\overline{Q}$ be the complement of $Q$ in $V$. Let $G_{Q} \defeq (Q, E_{Q})$ be a subgraph of $G$ where $E_{Q} \defeq \{(i,j) \in E| (i,j) \in Q\}$. Let $N_Q(i) \defeq \{j \in Q| (i,j) \in Q\}$.   
Let $\Gamma$ denote the product distribution that sets $\mathbb{P}_{\Gamma}(a_{ip}) = \pE[a_{ip}]$.  Let $\Gamma_Q$ be the product distribution defined over the nodes in $Q$ such that $\mathbb{P}_{\Gamma_Q}(a_{ip}) = \mathbb{P}_{\Gamma}(a_{ip})$ for all $i \in Q$. 

Let $\Psi_{\Gamma_Q}(\cdot): \overline{Q} \rightarrow \R^k$ denote a function taking as input a node $\overline{Q}$ and outputting a vector $v = (v_1,...,v_k) \in \R^k$.
\[
v_p \defeq 
\begin{cases} 
1 & \E_{\Gamma_Q}[\sum_{j \in N_{Q}(i)} \sum_{q \in [k]} f_{i,j}(p,q)a_{jq}] \geq \E_{\Gamma_Q}[\sum_{j \in N_{Q}(i)} \sum_{q \in [k]} f_{i,j}(p',q)a_{jq}] \text{ }\forall p' \in [k]\\
0 & \text{otherwise}
\end{cases}
\]
Furthermore let $\Lambda_{\Gamma_{Q}}(i) \defeq \frac{\Psi_{\Gamma_Q}(i)}{\| \Psi_{\Gamma_Q}(i)\|_1}$.  With a slight abuse of notation, we will sometimes interpret $\Lambda_{\Gamma_{Q}}(i)$ as a probability distribution over $[k]$ for all $i \in \overline{Q}$.  We also use $\Gamma_Q \bigcup \Lambda_{\Gamma_Q}(i)$ for $i \in \overline{Q}$ to denote appending the best response distribution over actions of node $i$ onto the distribution $\Gamma_Q$ defined over the nodes $Q$.  Similarly we use the notation $\Gamma_Q \bigcup_{i \in Y} \Lambda_{\Gamma_Q}(i)$ to denote appending the best response distribution of all nodes $i \in Y$ to $\Gamma_Q$ where $Y$ is any set of nodes disjoint from $Q$.

We begin our discussion with a few definitions.           
First we define $\Delta_\zeta(i)$ to be the average covariance of player $i$ with players that are neighbors of $i$.  
\begin{definition}
Let $\Delta_\zeta(i)$ for $i \in [N]$ be defined as 
\[\Delta_\zeta(i) \defeq \sum_{j \in N(i)}\sum_{(p,q) \in [k]^2}|\pE_\zeta[a_{ip}a_{jq}] - \pE_\zeta[a_{ip}]\pE_\zeta[a_{jq}]|\]
\end{definition}
Next we define $\Phi_\Gamma(i)$ for a product distribution $\Gamma \in \calP(k,V)$ to be the value of the Nash equilibrium constraint.
\begin{definition}
Let $\Gamma \in \calP(k,V)$ be a product distribution defined to be $\Gamma \defeq \{\pE_\zeta[a_{ip}]\}_{i \in [N], p \in [k]}$.  Then we define 
\[\Phi_{\Gamma}(i) \defeq \max_{w \in [k]} \Big(\sum_{j \in N(i)} \sum_{q \in [k]} f_{ij}(w,q) \pE_\zeta[a_{jq}] -\sum_{j \in N(i)} \sum_{(p,q) \in [k]^2} f_{ij}(p,q) \pE_\zeta[a_{ip}]\pE_\zeta[a_{jq}]\Big) \]
\end{definition}

Clearly we are trying to find a pseudoexpectation $\pE_\zeta$ such that $\Phi_\Gamma(i) \leq \epsilon$ for all $i \in [N]$.  The next lemma shows that after the rounding procedure in \pref{lem:rounding-main} we satisfy the constraints of $\Phi_\Gamma(i)$ up to an error that scales with the average pairwise covariance  
\begin{restatable}[]{lem}{constraintcorr}
\label{lem:constraintcorr}
Let $\Gamma \defeq \{\pE_\zeta[a_{ip}]\}_{i \in [N], p \in [k]}$ be a product distribution where $\pE_\zeta$ is the output of \pref{lem:rounding-main}.  Then we have that for all $i \in [N]$  
\[\Phi_{\Gamma}(i) \leq \epsilon + L\Delta_\zeta(i)\]
\end{restatable}
We defer the proof of \pref{lem:constraintcorr} to \pref{sec:brplemmas}.  With\pref{lem:constraintcorr}, it suffices to show that our rounding \pref{alg:brp} produces a product distribution such that $\Delta(i) \leq \frac{\epsilon}{L}$ for all $i \in [N]$.  

\begin{lemma}
If global correlation $\Delta \leq \epsilon^2$ then no more than $\epsilon N$ nodes satisfy $\Delta(i) \geq 2\epsilon$ 
\end{lemma}

\begin{proof}
Assume the contrary, if more than $\epsilon N$ nodes had $\Delta(i) \geq 2\epsilon$ then 
\[\Delta \defeq \frac{1}{|E|} \sum_{i=1}^N |N(i)| \Delta(i) \geq \frac{pN}{2p N^2} \sum_{i=1}^N  \Delta(i) \geq \epsilon^2\] 
which is a contradiction.  
\end{proof}
Finally, we need a lemma that says the Nash constraint is still $2\epsilon$ approximately satisfied if an $\frac{\epsilon}{L}$ fraction of the neighbors of $i$ play best response or in fact any response.
\begin{restatable}{lem}{smooth} 
\label{lem:smooth}
Let $Q$ be a subset of the vertices $V$ in graph $G$.  Let 
Let $\calP(k,Q)$ denote the space of product distributions over $[k]$ actions in $|Q|$ dimensions for agents in the vertices $Q$.  Let $\Gamma_Q$ be a product distribution in $\calP(k,Q)$.  Similarly, let $\calP(k,\overline{Q})$ denote the space of product distributions over $[k]$ actions in $|\overline{Q}|$ dimensions for agents in the vertices $\overline{Q}$.  
Then for all $i \in Q$ we have  
\[\max_{\Gamma_{\overline{Q}} \in \calP(k,\overline{Q})} \Phi_{\Gamma_Q \bigcup \Gamma_{\overline{Q}}}(i) \leq \Phi_{\Gamma_Q}(i) + L\frac{|N_{\overline{Q}}(i)|}{|N(i)|}\]
Here $\Phi_{\Gamma_Q \bigcup \Gamma_{\overline{Q}}}(i)$ is the value of the nash equilibrium constraint on agent $i$ over the vertices in $Q \bigcup \overline{Q} = V$ whereas $\Phi_{\Gamma_Q}(i)$ is only evaluated over with respect to the vertices $Q$.  
\end{restatable}
The proof is deferred to \pref{sec:brplemmas}.   The next subsection is devoted to the proof of  \pref{lem:best-response-pursuit}.  

\subsection{Best Response Rounding Proof}
In this section we prove \pref{lem:best-response-pursuit}, which provides guarantees for \pref{alg:brp}.  First we present the guarantees for \pref{alg:flag} 

\begin{restatable}{lem}{appendlemma} 
\label{lem:append-lemma} Let $H$ be a subset of $V$ with $|H| \geq ln^3(N)$.  Then $W = Flag(H,G)$ satisfies the following   

\begin{enumerate}
    \item $|W| \leq \frac{4}{3}|H|$
    \item $\frac{|N_{W}(i)|}{d} \leq \frac{40 |H|}{3N}+ 10\frac{\ln^2(N)}{d}$ for all $i \in \overline{W}$
\end{enumerate}

with high probability $1 - N^{-\ln^3(N)}\ln(N)$ over the randomness of $G$.  
\end{restatable} 

Next we present guarantees for \pref{alg:bestrespond}. 
\begin{restatable}{lem}{bestresponselemma} 
\label{lem:best-response}
Let $W \subseteq V$ and let $\Gamma \in \calP(k,N)$ be a product distribution over $k$ elements in $N$ dimensions.  We represent the product distribution $\Gamma \defeq \{\{\mu_{ip}\}_{p \in [k]}\}_{i \in V}$ where $\mu_{ip} = \mathbb{P}[a_{ip}=1]$.  Let $\Gamma_{\overline{W}} = \{\{\mu_{ip}\}_{p \in [k]}\}_{i \in \overline{W}}$ .  Let $\zeta \defeq \max_{i \in \overline{W}} \Phi_{\Gamma_{\overline{W}}}(i)$.  Let $\gamma \defeq \max_{i \in \overline{W}} \frac{|N_W(i)|}{d} $.  Then $BestRespond(W, \Gamma)$ outputs $(B, \hat{\Gamma})$ satisfying $|B| \leq \frac{|W|}{10}$ and $\max_{i \in \overline{B}} \Phi_{\hat{\Gamma}}(i) \leq \max(\zeta + L\gamma, 10 L\epsilon )$    
\end{restatable}
Finally we move on to prove \pref{lem:best-response-pursuit}
\bestresponsepursuit*

\begin{proof}
  Let $\zeta_t \defeq \max_{i \in \overline{W}} \Phi_{\Gamma_t}(i)$, and let $\gamma_t \defeq \max_{i \in \overline{W}} \frac{|N_W(i)|}{d}$.  For any $t \in [T-1]$ we have, 
\[\zeta_{t+1}\leq \max(\zeta_{t} + L\gamma_{t+1}, 100L\epsilon)
\leq \max(\zeta_{t} + L(\frac{4c}{3}|K_{t+1}| + c\frac{\ln^2(N)}{d}), 100L\epsilon)\]
Here the first inequality follows from \pref{lem:best-response}.  The second inequality follows from \pref{cor:append-lemma} which establishes that $\gamma_t \leq \frac{4c}{3}|K_t| + c \frac{\ln^2(N)}{d}$. Next, we continue to upper bound by   

\[\leq \max(\zeta_{t} + L (\frac{4c}{3} \frac{1}{4^{t}}\epsilon + c \frac{\ln^2(N)}{d}), 100L\epsilon) \leq \max(\zeta_0, 100L\epsilon) + L\frac{4c}{3} \sum_{i=1}^t \frac{1}{4^i} \epsilon + Lc \frac{t\ln^2(N)}{d}\]
 Here the first inequality follows from \pref{lem:best-response} where we have $|K_{t+1}| \leq \frac{|K_t|}{4}$.  The second inequality follows from the generic observation that for a sequence $a_1,a_2,...,a_{T} \in \R$ satisfying $a_{t+1} = \max(a_t + c_{t+1},b)$, we can conclude $a_{t+1} \leq \max(a_0,b) + \sum_{i=1}^{t+1} c_i$. 
 
 \[\leq 100L\epsilon + L\frac{16c}{9}\epsilon + Lc \frac{t\ln^2(N)}{d} \leq 100L\epsilon + L\frac{16c}{9}\epsilon + Lc \frac{\ln^3(N)}{d} \leq 100L\epsilon + L\frac{16c}{9}\epsilon + Lc\epsilon\] 
 
 The first inequality follows from the base case $\zeta_0 \leq 2\epsilon$.  The second inequality follows from $t \leq \ln(N)$, which holds because $K_t$ contracts exponentially until $K_{T-1} \leq \ln^3(N)$.  Note that the above bound is independent of $t$ because we have carried out all the geometric sums.  Therefore, $\zeta_T \leq 100L\epsilon + L\frac{16c}{9} \epsilon$ which guarantees for all $i \in V$ that 

\[\Phi_{\Gamma_T}(i) \leq 100L\epsilon + 20L\epsilon\] 
\end{proof}

\section*{Acknowledgements}
The author would like to thank Prasad Raghavendra, Ankur Moitra, and Costis Daskalakis for helpful discussions.
\newpage
\bibliographystyle{alpha}
\bibliography{bibliography}

\appendix

\section{Feasibility of Convex Hierarchy}

\begin{lemma}
Let $(G, \{f_{ij}\}_{i,j \in [N]})$ be a polymatrix game defined on graph $G$ with rescaled payoffs $\{f_{ij}\}_{i,j \in [N]}$ that are $L$-smooth.  If $G$ is a random graph with edge density $\frac{k^8}{\epsilon^4} \ln^3(N)$ we have the constraints of \pref{alg:convex} are feasible for $\ell = \frac{L^4\ln(k)}{\epsilon^4}$.  In particular, any mixed nash equilibrium of the game is a feasible pseudodistribution.  
\end{lemma}

\begin{proof}
Let $\zeta \in \calP(N,k)$ be the product distribution over $[k]^N$ corresponding to a mixed nash equilibrium.  Since $\zeta$ is a distribution we immediately satisfy $\pE_\zeta[1] = 1$ and $\pE_\zeta[p(\bold{a})] \geq 0$ for all $p(\bold{a}) \in SoS_{\ell}(\{a_{ir}\}_{i \in [N], r \in [k]})$. Furthermore, it is immediate that  $\pE[(a_{ip}^2 - a_{ip})\prod_{a_{uv} \in S}a_{uv}] = 0$  for all $i \in [N]$ and $p \in [k]$, and for all $S \subset \bold{a}$ of size $|S| \leq \ell - 2$.  It is also immediate that $\pE[(\sum_{p=1}^k a_{ip} - 1)\prod_{a_{uv} \in S}a_{uv}] = 0$ for all $i \in [N]$ and for all $S \subset \bold{a}$ of size $|S| \leq \ell - 1$.  Next we have 
\[\E_\zeta[(\sum_{j \in N(i)} \sum_{p,q \in [k]^2}f_{ij}(p,q)a_{ip}a_{jq} - \sum_{j \in N(i)} \sum_{q \in [k]}f_{ij}(w,q)a_{jq} + 2\epsilon)\prod_{a_{uv} \in S}a_{uv}] \]

Let $R \subseteq S$ where for all $a_{uv} \in R$ we have $u \in N(i)$.  This enables us to write    
\[= \E_\zeta[(\sum_{j \in N(i)} \sum_{p,q \in [k]^2}f_{ij}(p,q)a_{ip}a_{jq} - \sum_{j \in N(i)} \sum_{q \in [k]}f_{ij}(w,q)a_{jq} + 2\epsilon)\prod_{a_{uv} \in R}a_{uv}\prod_{a_{uv} \in S/R}a_{uv}] \]
Using the fact that $\zeta$ is a product distribution we obtain 

\begin{equation} \label{eq:feasible1}
= \E_\zeta[(\sum_{j \in N(i)} \sum_{p,q \in [k]^2}f_{ij}(p,q)a_{ip}a_{jq} - \sum_{j \in N(i)} \sum_{q \in [k]}f_{ij}(w,q)a_{jq} + 2\epsilon)\prod_{a_{uv} \in R}a_{uv}] \E_\zeta[\prod_{a_{uv} \in S/R}a_{uv}]     
\end{equation}
Since $\prod_{a_{uv} \in S/R}a_{uv}$ is a sum of squares, we have  $\E_\zeta[\prod_{a_{uv} \in S/R}a_{uv}] \geq 0$.  So to show \pref{eq:feasible1} is nonnegative, it suffices to show  

\begin{equation} \label{eq:feasible2}
\E_\zeta[(\sum_{j \in N(i)} \sum_{p,q \in [k]^2}f_{ij}(p,q)a_{ip}a_{jq} - \sum_{j \in N(i)} \sum_{q \in [k]}f_{ij}(w,q)a_{jq} + 2\epsilon)\prod_{a_{uv} \in R}a_{uv}] \geq 0    
\end{equation}
Moving on, we break $N(i)$ into $N(i)/R$ and $R$ and write \pref{eq:feasible2} accordingly, 
\begin{multline}
  \pref{eq:feasible2} = \E_\zeta[(\sum_{j \in N(i)/R} \sum_{p,q \in [k]^2}f_{ij}(p,q)a_{ip}a_{jq} - \sum_{j \in N(i)/R} \sum_{q \in [k]}f_{ij}(w,q)a_{jq})\\ + \sum_{j \in R} \sum_{(p,q) \in [k]^2}f_{ij}(p,q)a_{ip}a_{jq} - \sum_{j \in R}\sum_{q \in [k]}f_{ij}(w,q)a_{jq} + 2\epsilon) \prod_{a_{uv} \in R}a_{uv}]  
\end{multline}

\begin{multline}
  = \E_\zeta[\big(\sum_{j \in N(i)/R} \sum_{p,q \in [k]^2}f_{ij}(p,q)a_{ip}a_{jq} - \sum_{j \in N(i)/R} \sum_{q \in [k]}f_{ij}(w,q)a_{jq}\big)\prod_{a_{uv} \in R}a_{uv}\\ + \sum_{j \in R} \sum_{(p,q) \in [k]^2}f_{ij}(p,q)a_{ip}a_{jq}\prod_{a_{uv} \in R}a_{uv} - \sum_{j \in R}\sum_{q \in [k]}f_{ij}(w,q)a_{jq}\prod_{a_{uv} \in R}a_{uv} + 2\epsilon \prod_{a_{uv} \in R}a_{uv}]  
\end{multline}
By linearity 
\begin{multline}
  = \E_\zeta[\big(\sum_{j \in N(i)/R} \sum_{p,q \in [k]^2}f_{ij}(p,q)a_{ip}a_{jq} - \sum_{j \in N(i)/R} \sum_{q \in [k]}f_{ij}(w,q)a_{jq}\big)\prod_{a_{uv} \in R}a_{uv}]\\ + \pE_\zeta[\sum_{j \in R} \sum_{(p,q) \in [k]^2}f_{ij}(p,q)a_{ip}a_{jq}\prod_{a_{uv} \in R}a_{uv} - \sum_{j \in R}\sum_{q \in [k]}f_{ij}(w,q)a_{jq}\prod_{a_{uv} \in R}a_{uv} + 2\epsilon \prod_{a_{uv} \in R}a_{uv}]  
\end{multline}

Using the fact that $\zeta$ is a product distribution we obtain 
\begin{multline} \label{eq:feasible3}
  = \E_\zeta[\big(\sum_{j \in N(i)/R} \sum_{p,q \in [k]^2}f_{ij}(p,q)a_{ip}a_{jq} - \sum_{j \in N(i)/R} \sum_{q \in [k]}f_{ij}(w,q)a_{jq}\big)] \pE_\zeta[\prod_{a_{uv} \in R}a_{uv}]\\ + \pE_\zeta[\sum_{j \in R} \sum_{(p,q) \in [k]^2}f_{ij}(p,q)a_{ip}a_{jq}\prod_{a_{uv} \in R}a_{uv} - \sum_{j \in R}\sum_{q \in [k]}f_{ij}(w,q)a_{jq}\prod_{a_{uv} \in R}a_{uv} + 2\epsilon \prod_{a_{uv} \in R}a_{uv}]  
\end{multline}

Again we use the fact that $\prod_{a_{uv} \in R}a_{uv}$ is a sum of squares so $\pE_\zeta[\prod_{a_{uv} \in R}a_{uv}] \geq 0$.  We also use the fact that $\zeta$ is a mixed nash equilibrium so  

\[\E_\zeta[\big(\sum_{j \in N(i)/R} \sum_{p,q \in [k]^2}f_{ij}(p,q)a_{ip}a_{jq} - \sum_{j \in N(i)/R} \sum_{q \in [k]}f_{ij}(w,q)a_{jq}\big)] \geq -\epsilon \]

Plugging both these facts back into the first term of \pref{eq:feasible3} we obtain  

\begin{equation} \label{eq:feasible4} 
   \pref{eq:feasible3} \geq \pE_\zeta[\sum_{j \in R} \sum_{(p,q) \in [k]^2}f_{ij}(p,q)a_{ip}a_{jq}\prod_{a_{uv} \in R}a_{uv} - \sum_{j \in R}\sum_{q \in [k]}f_{ij}(w,q)a_{jq}\prod_{a_{uv} \in R}a_{uv} + \epsilon \prod_{a_{uv} \in R}a_{uv}]      
\end{equation}

By $L$-smoothness we have $f_{ij}(p,q) \geq -\frac{L}{d}$.  Therefore, 
the first term $\pE_\zeta[\sum_{j \in R} \sum_{(p,q) \in [k]^2}f_{ij}(p,q)a_{ip}a_{jq}\prod_{a_{uv} \in R}a_{uv}]$ is lower bounded by $\pE_\zeta[\frac{-L}{d}\sum_{j \in R} \sum_{(p,q) \in [k]^2}a_{ip}a_{jq}\prod_{a_{uv} \in R}a_{uv}]$.  Similarly, using the fact that $f_{ij}(p,q) \leq \frac{L}{d}$, we conclude that the second term $- \pE_\zeta[\sum_{j \in R}\sum_{q \in [k]}f_{ij}(w,q)a_{jq}\prod_{a_{uv} \in R}a_{uv}]$ is lower bounded by $- \pE_\zeta[\frac{L}{d}\sum_{j \in R}\sum_{q \in [k]}a_{jq}\prod_{a_{uv} \in R}a_{uv}]$.  Plugging both relations into \pref{eq:feasible4} we obtain 

\begin{equation} 
   \geq  \pE_\zeta[\frac{-L}{d}\sum_{j \in R} \sum_{(p,q) \in [k]^2}a_{ip}a_{jq}\prod_{a_{uv} \in R}a_{uv} - \frac{L}{d}\sum_{j \in R}\sum_{q \in [k]}a_{jq}\prod_{a_{uv} \in R}a_{uv} + \epsilon \prod_{a_{uv} \in R}a_{uv}]
\end{equation}

\begin{equation} 
   =  \pE_\zeta[\big(\frac{-L}{d}\sum_{j \in R} \sum_{(p,q) \in [k]^2}a_{ip}a_{jq} - \frac{L}{d}\sum_{j \in R}\sum_{q \in [k]}a_{jq} + \epsilon\big) \prod_{a_{uv} \in R}a_{uv}]
\end{equation}
Using the booleanity constraint we lower bound by 
\begin{equation} \label{eq:feasible5}
   \geq \pE_\zeta[\big(-\frac{L|R|(k^2 + k)}{d} + \epsilon\big) \prod_{a_{uv} \in R}a_{uv}] 
\end{equation}
   
Plugging $d = \frac{Lk^8 \ln^3(N)}{\epsilon^5}$ and $|R| = \frac{\ln(k)}{\epsilon^4}$ we obtain 
\[
   -\frac{L|R|(k^2 + k)}{d} + \epsilon \geq 0\]
Since $\prod_{a_{uv} \in R}a_{uv}$ is a sum of squares we conclude $\pref{eq:feasible5} \geq 0$ as desired.      
TODO: Change $d = \frac{Lk^8 \ln^3(N)}{\epsilon^5}$
\end{proof}

\section{Explicit Convex Constraints} \label{sec:explicit}

\paragraph{Moment Matrices: }

Let $M_{\ell}(\{a_{ir}\}_{i \in [N], r \in [k]})$ be the moment matrix with rows indexed by subsets $I \subseteq \{a_{ir}\}_{i \in [N], r \in [k]}$ for $|I| \leq \ell$.  Similarly, columns are indexed by subsets $J \subseteq \{a_{ir}\}_{i \in [N], r \in [k]}$ for $|J| \leq \ell$.
We require that 
\[M_{\ell}(\bold{a}) \succeq 0\]
\[M_{\ell}(\bold{a})_{\emptyset,\emptyset} = 1\]
Note that the PSD constraint implies via the cholesky decomposition that
\[M_{\ell}(\bold{a}) = V^TV\]
for a matrix $V \in \R^{{Nk \choose \ell} \times {Nk \choose \ell}}$ with columns equal to $v_H$ for all $H \subset \{a_{ir}\}_{i \in [N], r \in [k]}$ satisfying $|H| \leq \ell$.  Now, given the vectors   $\{v_H\}_{H \subset \{a_{ir}\}_{i \in [N], r \in [k]}, |H| \leq \ell}$ we add the linear constraints that 
\[\pE[\prod_{a_{uv} \in H} a_{uv}] \defeq \langle v_{H_1}, v_{H_2} \rangle \]
for all $H_1$ and $H_2$ satisfying $H_1 \cup H_2 = H$.  Henceforth we may refer to $\pE[\prod_{(u,v) \in H} a_{uv}]$ without ambiguity.  Also note that we have satisfied the $\pE[1] = 1$ and $\pE[p(\bold{a})] \geq 0$ constraints for all $p(\bold{a}) \in SoS_{\ell}(\bold{a})$.  At this point, we can define the linear functional $\pE: \R_{\leq \ell}[\bold{a}] \rightarrow \R$ from the degree less than $\ell$ polynomials to the reals so that $\pE[\sum_{i} \alpha_i p_i(a)] = \sum_{i} \alpha_i \pE[p_i(a)]$.       

The booleanity and color constraints are relaxed in the manner prescribed by the Sherali-Adams hierarchy.  We add linear constraints for all $ i \in [N]$ and $p \in [k]$

\[ \pE[(a_{ip}^2 - a_{ip})\prod_{a_{uv} \in S}a_{uv}] = \pE[a_{ip}^2\prod_{a_{uv} \in S}a_{uv}] - \pE[a_{ip}\prod_{a_{uv} \in S}a_{uv}] = 0\]
for all $S \subset \{a_{ir}\}_{i \in [N], r \in [k]}$ with size $|S| \leq \ell - 2$.  Additionally, we add linear constraints for all $i \in [N]$

\[\pE[(\sum_{r=1}^k a_{ir} - 1)\prod_{a_{uv} \in S}a_{uv}] = \sum_{r \in [k]}\pE[ a_{ir}\prod_{a_{uv} \in S}a_{uv}] - \pE[\prod_{a_{uv} \in S}a_{uv}] = 0\]
for all $S \subset \{a_{ir}\}_{i \in [N], r \in [k]}$ satisfying $|S| \leq \ell - 1$.  Finally, we add the linear constraints for conditioning preserving correlated equilibrium.  For all $i \in [N]$ and for all $w \in [k]$

\[\pE[\big(\sum_{j \in N(i)} \sum_{p,q \in [k]^2}f_{ij}(p,q)a_{ip}a_{jq}\big)\prod_{a_{uv} \in S}a_{uv}] - \pE[\big(\sum_{j \in N(i)} \sum_{q \in [k]}f_{ij}(w,q)a_{jq} \big)\prod_{a_{uv} \in S}a_{uv}] + \epsilon \geq 0\]

for all $S \subset \{a_{jr}\}_{j \in [N]/i, r \in [k]}$ with size $|S| \leq \ell - 2$.

\section{Best Response Pursuit Lemmas} \label{sec:brplemmas}
\constraintcorr*

\begin{proof}
We begin with the definition 
\[\Phi_{\Gamma}(i) = \max_{w \in [k]}\Big(\sum_{j \in N(i)} \sum_{q \in [k]} f_{ij}(w,q) \pE[a_{jq}] - \sum_{j \in N(i)} \sum_{(p,q) \in [k]^2} f_{ij}(p,q) \pE[a_{ip}]\pE[a_{jq}]\Big) \]

Let $w^*$ be the argmax over $w \in [k]$ in the equation above.  Then we rewrite the equation as 

\begin{multline*}
= \sum_{j \in N(i)} \sum_{(p,q) \in [k]^2} f_{ij}(p,q) (\pE[a_{ip}a_{jq}] - \pE[a_{ip}]\pE[a_{jq}])\\  + \sum_{j \in N(i)} \sum_{q \in [k]} f_{ij}(w^*,q) \pE[a_{jq}] - \sum_{j \in N(i)} \sum_{(p,q) \in [k]^2} f_{ij}(w^*,q) \pE[a_{ip}a_{jq}]     
\end{multline*}

The second line is the conditioning preserving correlated equilibrium constraint which we know is less than $\epsilon$.  Therefore we upper bound by  
\[\leq \epsilon + \sum_{j \in N(i)} \sum_{(p,q) \in [k]^2} f_{ij}(p,q) (\pE[a_{ip} a_{jq}]) - \pE[a_{ip}]\pE[a_{jq}] ) \]
We upper bound via $\ell_{1,\infty}$ holders and use the fact that $f_{i,j}(a,b) \leq \frac{L}{d}$.   
\[\leq \epsilon + \frac{L}{d}\sum_{j \in N(i)} \sum_{(p,q) \in [k]^2} |\pE[a_{ip} a_{jq}] - \pE[a_{ip}]\pE[a_{jq}]|]  = \epsilon + L\Delta(i)\]
\end{proof}

\smooth*

\begin{proof}
We have by definition 

\[\Phi_{\Gamma_Q \bigcup \Gamma_{\overline{Q}}}(i) \defeq \max_{w \in [k]} \E_{\Gamma_Q \bigcup \Gamma_{\overline{Q}}}[ \sum_{j \in N_{Q \bigcup \overline{Q}}(i)}\sum_{q \in [k]} f_{i,j}(w, q)a_{jq} - \sum_{j \in N_{Q \bigcup \overline{Q}}(i)}\sum_{(p,q) \in [k]^2} f_{i,j}(p,q)a_{ip}a_{jq}]\]

Using the fact that $N_{Q \bigcup \overline{Q}}(i)  = N_Q(i) \bigcup N_{\overline{Q}}(i)$ we obtain 

\begin{multline}
     = \max_{r \in [k]} \E_{\Gamma_Q \bigcup \Gamma_{\overline{Q}}}[ \sum_{j \in N_Q(i)}\sum_{p \in [k]} f_{i,j}(w, p)a_{jp} + \sum_{j \in N_{\overline{Q}}(i)}\sum_{p \in [k]} f_{i,j}(w,p) a_{jp}\\ - (\sum_{j \in N_Q(i)}\sum_{(p,q) \in [k]^2} f_{i,j}(p,q)a_{ip}a_{jq} + \sum_{j \in N_{\overline{Q}}(i)}\sum_{(p,q) \in [k]^2} f_{i,j}(p,q)a_{ip}a_{jq})]
\end{multline}

Using the $L$-smooth assumption we obtain 

\[ \leq \max_{w \in [k]} \E_{ \Gamma_Q \bigcup \Gamma_{\overline{Q}}}[ \sum_{j \in N_Q(i)} \sum_{p \in [k]}f_{i,j}(w, p)a_{jp} - \sum_{j \in N_Q(i)} \sum_{(p,q) \in [k]^2}f_{i,j}(p,q)a_{ip}a_{jq} + 2L\frac{|N_{\overline{Q}}(i)|}{|N(i)|}]\]

\[ = \max_{w \in [k]} \E_{\Gamma_Q}[ \sum_{j \in N_Q(i)}\sum_{p \in [k]} f_{i,j}(w,p) a_{jp} - \sum_{j \in N_Q(i)}\sum_{(p,q)\in [k]^2} f_{i,j}(p,q)a_{ip}a_{jq}] + 2L\frac{|N_{\overline{Q}}(i)|}{|N(i)|}\]
Using the definition of $\Phi_Q(i)$ we conclude 
\[ = \Phi_Q(i) + 2L\frac{|N_{\overline{Q}}(i)|}{|N(i)|}\]
as desired.  
\end{proof}

To prove \pref{lem:append-lemma} we will need the following supporting lemma.  
\begin{lemma} \label{lem:append-exponent}
For any subset of nodes $H \subset [N]$, the number of nodes $i \in \overline{H}$ satisfying $|N_H(i)| \geq \max(10 \frac{|H|}{N} d, 10\ln(n))$ is less than $\frac{H}{4}$.    
\end{lemma}

\begin{proof}
We proceed by an application of the chernoff and union bound.  Let's fix the set $H$.  Then for any node $i \in [N]$, the probability $|N_H(i)| \geq \max(10 \frac{|H|}{N} |E(i)|, 10\ln(n))$ is upper bounded as follows.  Let $e_j$ be the following random variable with $e_j = 1$ if $j \in H$ and $e_j = 0$ otherwise.  Here the randomness is taken with respect to the randomness in generating the graph.  Let $\mu \defeq \E[\frac{1}{N}\sum_{j=1}^N e_j] = \frac{|H|}{N}\frac{d}{N}$.  We define a constant $\delta > 0$ such that $(1+\delta)\mu N = \max(10 \frac{|H|}{N} d, 10\ln(n)) \geq 10 \frac{|H|}{N} d$.  This implies $\delta > 9$. Now let $\tau \defeq \mathbb{P}[\frac{1}{N}\sum_{j=1}^N e_j \geq (1+\delta)\mu]$.  Using the chernoff bound we obtain.   
\[\tau \defeq \mathbb{P}[\frac{1}{N}\sum_{j=1}^N e_j \geq (1+\delta)\mu] \leq \exp(-\delta \mu N) \leq \exp(-9\mu N)\]
If $\mu N \geq \ln(N)$ then we have $\tau \leq \exp(-9 \ln(N))$.  On the other hand, if $\mu N \leq \ln(N)$ we define $\delta$ such that 
$(1+\delta)\mu N = \max(10 \frac{|H|}{N} d, 10\ln(n)) \geq 10 \ln(N)$ which implies $\delta \mu N \geq 9 \ln(N)$.  Thus we have in this case 

\[\tau \defeq \mathbb{P}[\frac{1}{N}\sum_{j=1}^N e_j \geq (1+\delta)\mu] \leq \exp(-\delta \mu N) \leq \exp(-9\ln(N))\]
Now let $H^* \defeq \{i \in \overline{H}| |N_H(i)| \geq \max(10 \frac{|H|}{N} d, 10\ln(n))\}$.  We need to show $|H^*| \leq \frac{|H|}{4}$.  Let $w_i$ for $i \in [N]$ be a random variable where $w_i = 1$ if $i \in H^*$ and $w_i = 0$ otherwise. Let $\mu' \defeq \mathbb{E}[\frac{1}{N}\sum_{i=1}^N w_i] = \tau$.  Let $\delta' > 0$ be defined such that $(1+\delta')\mu' N = \frac{|H|}{4}$ then by chernoff bound we have
\[\mathbb{P}(\frac{1}{N}\sum_{i=1}^N w_i \geq \mu'(1+\delta')) \leq \Big(\frac{e^{\delta'}}{(1+\delta')^{(1+\delta')}}\Big)^{\mu'N} = \exp((\delta'-(1+\delta') \ln(1+\delta')) \mu' N)  
\]
Using the fact that $(1+\delta')\mu' N = \frac{|H|}{4}$ we obtain
\begin{equation}
\leq \exp(\frac{|H|}{4} -\ln(1+\delta') \frac{|H|}{4})  
\end{equation}
Next we use the fact that $\delta' = \frac{|H|}{4\mu' N} - 1 \geq \frac{1}{4N} \exp(9 \ln(N)) - 1 \geq \frac{N^8}{4} - 1$.  Plugging $\delta' \geq \frac{N^8}{4} - 1$ into 
\[\leq \exp(-8\ln(N) \frac{|H|}{4}) \leq N^{-2|H|}    \]
Since there are no more than ${N \choose |H|} \leq N^{|H|}$ subsets of size $|H|$  we conclude via union bound that  $|H^*| \leq \frac{|H|}{4}$ with probability $ 1 - N^{-2|H|} N^{|H|} \geq  1 - N^{-|H|} \geq 1 - N^{-\ln^3(N)}$ as desired.

\end{proof}

\appendlemma*

\begin{proof}
Let $H_1 \defeq H$.  Then we have by \pref{lem:append-exponent} that $Flag(H_1,G)$ iteratively appends $H_2,...,H_y$ to $H_1$ where $H_i \leq \frac{1}{4^{j-1}}H_1$ for all $j \in [2,y]$.  Thus $|W| \leq \frac{4}{3}|H|$.  Furthermore, each iteration of of $Flag(H_1,G)$ ensures that for all $i \in \overline{W}$ 
\begin{equation} \label{eq:app-exp-1}
\frac{|N_{W}(i)|}{d} = \frac{1}{d}|N_{\bigcup_{j \in [y]} H_j}(i)| = \sum_{j=1}^y \frac{|N_{H_j}(i)|}{d} \leq \sum_{j=1}^y 10\max \Big(\frac{|H_j|}{N}, \frac{\ln(N)}{d}\Big)    
\end{equation}

\[\leq 10\sum_{j=1}^y \Big(\frac{|H_j|}{N} + \frac{\ln(N)}{d}\Big)\leq \frac{40}{3N}|H_1| + 10\frac{\ln^2(N)}{d}\] 

Here the first inequality follows from the fact that at each iteration $Flag(H_1,G)$ performs the following operation.  For $W = \bigcup_{r=1}^j H_r$ we define $H_{j+1}$ to be 
\[H_{j+1} \defeq \Big\{i \in \overline{W}\Big| \frac{|N_{H_{j}}(i)|}{d} \geq \max \Big(10 \frac{|H_{j}|}{N}, 10\frac{\ln(n)}{d}\Big)\Big\}\]
Thus for all $i \in \overline{W}$ we have 

\[\frac{|N_{H_{j}}(i)|}{d} \leq \max \Big(10 \frac{|H_{j}|}{N}, 10\frac{\ln(n)}{d}\Big)\]
Continuing to bound \pref{eq:app-exp-1}, we separate the max into a sum to conclude
\[\leq 10\sum_{j=1}^y \Big(\frac{|H_j|}{N} + \frac{\ln(N)}{d}\Big)\leq \frac{40}{3N}|H_1| + 10\frac{\ln^2(N)}{d}\] 
We have established $\frac{|N_{W}(i)|}{d} \leq \frac{40 |H|}{3N}+ 10\frac{\ln^2(N)}{d}$ for all $i \in \overline{W}$ as desired.  
\end{proof}

\bestresponselemma*

\begin{proof}

First, using \pref{lem:smooth} we have $\max_{i \in \overline{W}} \Phi_{\hat{\Gamma}}(i) \leq \zeta + L\gamma$ which concludes the lemma for $i \in \overline{W}$.  To conclude the lemma it suffices to show that $BestRespond(W,\Gamma)$ sets a subset $A \subseteq W$ to best respond to $\Gamma_{\overline{W}}$, and that for all $i \in A$, $\frac{|N_W(i)|}{d} \leq 10\epsilon$.  

To show this, let $\kappa \in \R^+$ be a constant such that $|W| = \kappa N$.  Let $e_1,...,e_{\kappa^2N^2}$ be the indicator random variables distributed $Bern(p)$ for the presence of an edge between any pair of nodes in $W$.  Let $\mu = p$, and let $\delta \mu = \frac{10\epsilon d}{\kappa N} = \frac{10\epsilon}{\kappa} p = \frac{10\epsilon}{\kappa} \mu$ which implies $\delta = \frac{10\epsilon}{\kappa}$.  Then for $|W| \leq \epsilon N$ we have $\delta \geq 10$.  Therefore, we invoke the upper chernoff bound to obtain  
\[\mathbb{P}[\frac{1}{\kappa^2 N^2} \sum_{i=1}^{\kappa^2 N^2} e_i \geq \frac{10\epsilon d}{\kappa N} ] \leq \exp(-\frac{9\epsilon d}{\kappa N} \kappa^2 N^2) = \exp(-9\epsilon d \kappa N) = \exp(-9 \frac{1}{\epsilon^3}\ln^3(N) \kappa N)\]

We also know there are no more than ${N \choose \kappa N} \leq N^{\kappa N}$ subsets of $V$ of size $\kappa N$.  Therefore, applying union bound we conclude that $W$ has fewer than $10 \epsilon d \kappa N$ edges with high probability $1 - N^{-\ln^3(N)}$.  Furthermore, via markov we conclude that no more than $\frac{1}{10}|W|$ nodes $i \in W$ satisfy $N_W(i) \geq 100 \epsilon d$.  Equivalently, there exists a subset $A \subseteq W$ of size $|A| \geq \frac{9|W|}{10}$ such that $\frac{|N_W(i)|}{d} \leq 100 \epsilon$ for all $i \in A$. By \pref{lem:smooth}, we conclude $\max_{i \in A} \Phi_{\hat{\Gamma}}(i) \leq 100L\epsilon$.  

Putting everything together, for $BestResponsd(W,\Gamma)$ outputs $B = W/A$ such that  $\overline{B} = \overline{W} \bigcup A$.  Thus we conclude,
$$\max_{i \in \overline{B}} \Phi_{\hat{\Gamma}}(i) = \max_{i \in \overline{W} \bigcup A} \Phi_{\hat{\Gamma}}(i) = \max(\max_{i \in \overline{W}} \Phi_{\hat{\Gamma}}(i), \max_{i \in A} \Phi_{\hat{\Gamma}}(i)) \leq \max(\zeta + L\gamma, 100 L\epsilon )$$  
as desired.  

\end{proof}

\section{Miscellaneous}
\begin{lemma} \label{lem:jackson}(Bernstein, Jackson, folklore) 
There exists a series of polynomials $\{g_p(x)\}_{p=1}^n$ where  $g_{p}(x) \defeq \sum_{r=1}^{p} \alpha_r x^r$ and $\{\alpha_r\}_{r=1}^n$ is a series of coefficients satisfying  $\alpha_{r} \leq \frac{18}{\sqrt{\pi}} r^{-3/2}$ and $|g_n(x) - |x|| \leq \frac{6}{n}$.  
\end{lemma}

\begin{lemma}[\cite{RT12}] \label{lem:global_correlation} 
There exists $t \leq k$ such that $E_{i_1,...,i_t \sim W}E_{i,j \sim W}[I(X_i,X_j|X_{i_1}, ..., X_{i_t})] \leq \frac{\ln(q)}{k-1}$ where $q$ is alphabet size.  
\end{lemma}

\begin{proof}
Linearity of expectation we have for any $t \leq k-2$

\[E_{i,i_1,...i_t\sim W} [H(X_i| X_{i_1},..., X_{i_t})] =E_{i,i_1,...i_t\sim W} [H(X_i| X_{i_1},..., X_{i_{t-1}}) - E_{i_1,...,i_{t-1} \sim W}E_{i,i_t \sim W}[I(X_i,X_{i_t}|X_{i_1}, ..., X_{i_{t-1}})] \]

adding equalities from $t=1 $ to $k-2$ we get 

\[E[H(X_i)] - E_{i_1,...,i_{k-2} \sim W}[H(X_i|X_{i_1},...,X_{i_{k-2}})] = \sum_{i\leq t\leq k-1} E_{i,j,X_{i_1},..., X_{i_{t-1}} \sim W}[I(X_i,X_j|X_{i_1}, ..., X_{i_t})]\]

Lemma follows from $H(X) \leq \ln(q)$
\end{proof}

\begin{definition} (Payoff Rescaling) For each game $p_{ij}$ for $(i,j) \in E$, let 
\[U_i \defeq \max_{a_1,a_2,...,a_N \in [k]} \sum_{j \in N(i)} p_{ij}(a_i, a_j)\]

and let 

\[L_i \defeq \min_{a_1, a_2,...,a_N \in [k]} \sum_{j \in N(i)}p_{ij}(a_i, a_j)\]

Then let $f_{ij}(a_i,a_j) = (p_{ij}(a_i,a_j) - \frac{L_i}{|N(i)|}) \frac{1}{U_i - L_i}$ such that $\sum_{j \in N(i)} f_{ij}(a_i,a_j) \in [0,1]$ for all $\vec{a} \in [k]^N$.    
\end{definition}

\begin{fact} Let $x_1,...,x_N$ be i.i.d bernoulli random variables with $\mu = \E[x_i]$.    
Upper chernoff bound, for any $\delta > 0$  

\[\mathbb{P}(\frac{1}{N}\sum_{i=1}^N x_i \geq (1+\delta)\mu) \leq \Big(\frac{e^\delta}{(1+\delta)^{(1+\delta)}}\Big)^{\mu N}\]
Lower chernoff bound.  For $\delta \in [0,1]$ 
\[\mathbb{P}(\frac{1}{N}\sum_{i=1}^N x_i \geq (1-\delta)\mu) \leq \exp\Big(\frac{-\delta^2 \mu N}{2}\Big)\]
\end{fact}

\end{document}